\documentclass[11pt]{amsart}

\usepackage{fullpage}
\usepackage{latexsym}
\usepackage{amssymb,amsfonts,amsmath,amsthm}

\newcommand{\CQ}{QC}
\newcommand{\PG}{P\'erez-Garc\'ia et al. }
\newcommand{\GIP}{\mathrm{GIP}}

\newcommand{\eps}{\varepsilon}

\newcommand{\ket}[1]{|#1\rangle}
\newcommand{\bra}[1]{\langle#1|}
\newcommand{\braket}[2]{\langle #1|#2\rangle}

\newcommand{\Tr}{\mbox{\rm Tr}}

\newcommand{\R}{\mathbb{R}}
\newcommand{\C}{\mathbb{C}}

\newcommand{\K}{\mathbb K}
\newcommand{\pmset}[1]{\{\pm1\}^{#1}}

\newcommand{\beq}{\begin{equation}}
\newcommand{\eeq}{\end{equation}}
\newcommand{\beqn}{\begin{equation*}}
\newcommand{\eeqn}{\end{equation*}}
\newcommand{\beqr}{\begin{eqnarray}}
\newcommand{\eeqr}{\end{eqnarray}}
\newcommand{\beqrn}{\begin{eqnarray*}}
\newcommand{\eeqrn}{\end{eqnarray*}}

\DeclareMathOperator{\sign}{sign}
\DeclareMathOperator{\Ball}{\mathcal{S}}
\usepackage{amsfonts}
\usepackage{amssymb}
\usepackage{latexsym}
\usepackage{amsmath}
\usepackage{mathrsfs}
\newcommand{\HS}{\mathcal{H}}
\newcommand{\N}{\mathbb{N}}

\newcommand{\X}{\mathcal{X}}

\newcommand{\psiket}{\ket{\psi}}
\newcommand{\psibra}{\bra{\psi}}

\newcommand{\betag}{\beta^*_\textsc{Z}}
\newcommand{\betad}{\beta^*_\textsc{S}}
\newcommand{\betat}{\beta^*_\textsc{C}}

\newtheorem{definition}{Definition}
\newtheorem{theorem}{Theorem}
\newtheorem{lemma}[theorem]{Lemma}
\newtheorem{corollary}[theorem]{Corollary}
\newtheorem{claim}[theorem]{Claim}

\newtheorem{proposition}[theorem]{Proposition}

\newtheorem{remark}{Remark}

\renewenvironment{proof}[1][]{\begin{trivlist}
\item[\hspace{\labelsep}{\bf\noindent Proof#1:\/}] }{\qed\end{trivlist}}

\begin{document}

\begin{titlepage}
\begin{abstract}
XOR games are a simple computational model with connections to many areas of complexity theory including 
interactive proof systems, hardness of approximation, and communication complexity.  Perhaps the earliest use of 
XOR games was in the study of quantum correlations, as an experimentally realizable setup which could demonstrate 
non-local effects as predicted by quantum mechanics.  XOR games also have an interesting connection to Grothendieck's 
inequality, a fundamental theorem of analysis---Grothendieck's inequality shows that two players sharing 
entanglement can achieve at most a constant factor advantage over players following classical strategies in an XOR game.  

The case of multiplayer XOR games is much less well understood. \PG show the existence of entangled states which allow an
unbounded advantage for players in a three-party XOR game over their classical counterparts.  On the other hand, they show 
that when the players share GHZ states, a well studied multiparty entangled state, this advantage is bounded by a constant.

We use a multilinear generalization of Grothendieck's inequality due to Blei and Tonge to simplify the proof of the second result and 
extend it to the case of so-called Schmidt states, answering an open problem of \PG Via a reduction given in that paper, this 
answers a 35-year-old problem in operator algebras due to Varopoulos, showing that the space of compact operators on a Hilbert 
space is a Q-algebra under Schur product.  

A further generalization of Grothendieck's inequality due to Carne
lets us show that the gap between the entangled and classical value is
at most a constant in any multiplayer XOR game in which the players are
allowed to share combinations of GHZ states and EPR pairs of any
dimension.
Based on a result by Bravyi et al. this implies that in a three-party
XOR game, players sharing an arbitrary stabilizer state cannot achieve
more than a constant factor advantage over unentangled players.

Finally, we discuss applications of our results to communication complexity.  We show that the discrepancy method 
in communication complexity remains a lower bound in the multiparty model where the players have quantum communication 
and any of the kinds of entanglement discussed above.  This answers an open question of Lee, Schechtman, and Shraibman 
who showed that discrepancy was a lower bound on multiparty communication complexity but were unable to handle the case 
of entanglement.
\end{abstract}

\title{Multiplayer XOR games and quantum communication complexity with clique-wise entanglement}

\author{Jop Bri\"{et}}
\email[J. Bri\"{et}]{J.Briet@cwi.nl}
\thanks{CWI and University of Amsterdam.}
\author{Harry Buhrman}
\email[H. Buhrman]{buhrman@cwi.nl}
\thanks{CWI and University of Amsterdam.}
\author{Troy Lee}
\email[T. Lee]{troyjlee@gmail.com}
\thanks{Rutgers University}
\author{Thomas Vidick}
\email[T. Vidick]{vidick@cs.berkeley.edu}
\thanks{Computer Science Division, University of California, Berkeley.}

\date{\today}\maketitle

\end{titlepage}

\section{Introduction}

In an XOR game $G=(f,\pi)$, with probability $\pi(x,y)$ two parties Alice and Bob are given inputs 
$x$ and $y$ respectively.  Without communicating, their task is for Alice to output a bit $a$, and Bob a bit $b$, such that 
$a \oplus b=f(x,y)$.  Performance in an XOR game $(f,\pi)$ is usually measured as the largest {\em bias} 
$\beta(G)$ achievable by a protocol, where the bias of a protocol is the probability under $\pi$ that the output 
of the protocol is correct minus the probability under $\pi$ that the protocol is incorrect.

XOR games are a simple and natural model of computation and have been studied in the context of interactive 
proofs (the complexity class $\oplus$MIP), hardness of approximation (results for the simplest CSPs, such as MAX2SAT~\cite{Hastad:502098}, can be phrased as results on the inapproximability of the bias of XOR games), and communication complexity (where they are tightly related to the discrepancy method).  Perhaps more importantly, they have also proved to be an excellent framework in which to study the relationship between classical and quantum computation. For this, one can consider allowing Alice and Bob to share various resources, such as an entangled state or another non-local resource, in order to help them win the game by better coordinating their answers. As such, XOR games provide the simplest setting in which the non-local properties of quantum mechanics manifest themselves.

The case of two-player XOR games is reasonably well understood~\cite{Cleve:04a}. In this setting, Grothendieck's inequality, a fundamental inequality in Banach-space theory, plays a surprising role: in conjunction with Tsirelson's characterization \cite{tsirelson93:_some} of entangled 
XOR games, it essentially says that, in any XOR game, the entangled bias $\beta^*(G)$, which is defined as the maximum bias achievable by two players sharing any entangled state, can be at most a constant factor larger than the classical bias $\beta(G)$.  
For convenience, we will refer to the ratio $\beta^*(G)/\beta(G)$ as the \CQ-gap.

The case of entangled XOR games with more than two players is much less well understood. This is in part a reflection of the fact that multipartite entanglement has proved to be much more unwieldy than its bipartite counterpart. Among the few results known in this setting, one of the most striking is due to P\'erez-Garc\'ia et al. \cite{perezgarcia:2008}, who show that for every~$n$, there is a game~$G_n$ such that $\beta^*(G_n)  \geq C\sqrt{n}\, \beta(G_n)$ for some constant $C$.  Here $n$ refers to the smallest 
local dimension of the entangled state for any of the three players.  On the other hand, \PG also show that if the entanglement of the players is restricted to a GHZ state, the \CQ-gap is at most a constant depending only on the number of players. 
These results raise the following question: what kinds of entanglement allow for unbounded \CQ-gaps?

In this paper, we initiate the systematic study of multiplayer XOR games in which the players are allowed to share specific patterns of entanglement. The patterns that we  consider are quite general, and in fact include many of the states known to be easily prepared in the lab, such as stabilizer states. We now describe our results and techniques, together with their applications.

\subsection{Our results}
We study two main types of $N$-partite entanglement. The first is a generalization of GHZ states which we call \emph{Schmidt states}: states of the form $\ket{\Psi} = \sum_i \alpha_i \ket{i}^{\otimes N}$. The second is formed of any combination of EPR pairs or GHZ states shared among any subsets of the~$N$ parties. We call such entanglement \emph{clique-wise entanglement}. This is quite a general form of entanglement, as for instance it includes $3$-partite stabilizer states~\cite{bravyi:2006}. We denote by $\betad(G)$ (resp. $\betat(G)$) the maximal bias achievable in game $G$ by players who are restricted to sharing a Schmidt state of arbitrary dimension (resp. arbitrary clique-wise entanglement). Our results can be seen as proving constant upper bounds on the \CQ-gaps of these quantities. 

The main technical contribution of this paper is the expansion of Tsirelson's connection between \CQ-gaps of two-player XOR games and Grothendieck's inequality, to connections between \CQ-gaps for $N$-player XOR games with the above mentioned patterns of entanglement and certain multi-linear extensions of  Grothendieck's inequality.
The unifying idea underlying the proofs of our results is that, if the players share a specific type of entangled state, then the \CQ-gap can be upper-bounded by a constant appearing in a related Grothendieck-type inequality.

We now explain our results in more detail.  Concerning Schmidt states we prove the following theorem.
\begin{theorem}\label{schmidtbound}
Let $G$ be an $N$-player game. Then the maximum bias achievable by players sharing a Schmidt state $\ket{\Psi} := \sum_{i=1}^d\alpha_i\ket{i}^{\otimes N}$, for an arbitrary dimension $d$, is at most a constant times the classical bias. Formally:

$$\betad(G)\,\leq\, 2^{(3N-5)/2}\,K_G^{\C}\,\beta(G),$$
where $K_G^{\C}\lesssim 1.40491$ is the complex Grothendieck constant.
\end{theorem}
This generalizes -- with slightly improved constants -- a result of \PG who show a constant \CQ-gap for the case of GHZ states.   The main tool in this result is an extension of 
Grothendieck's inequality due to Blei \cite{blei:1979} and later simplified and improved by 
Tonge \cite{tonge:1978} while studying certain extensions of the von Neumann inequality.
The exponential dependence on the number of players is necessary in this theorem as
Zukowski \cite{zukowski:1993} has given an explicit sequence of $N$-player XOR games 
$G_N$ where players sharing a GHZ state can a achieve a bias $2^{-1}(\pi/2)^N$
times that of the classical bias.

Theorem~\ref{schmidtbound} answers an open question of P\'erez-Garc\'ia et al., who were 
particularly interested in 
the case of Schmidt states because of a connection to a 35-year-old open problem of Varopoulos 
on operator algebras.  
We are able to resolve this question via a reduction given in P\'erez-Garc\'ia et al.  This is 
described below in Section~\ref{intro:varopoulos}. 

Our second result further exploits the connection between \CQ-gaps and Grothendieck-type 
inequalities and deals with the case where the players share clique-wise entanglement. Here, we 
consider the general setting where the $N$ players are organized in $k$ coalitions of $r$ players 
each (a given player can take part in any number of coalitions). The members of each of the 
coalitions are allowed to share a GHZ state of arbitrary dimension among themselves. We relate 
the \CQ-gap in this setting to an inequality initially proved by Carne~\cite{carne:1980} in the 
context of Banach lattices. 
This lets us prove that, even in this complex setting, the \CQ-gap is bounded by a constant depending only on the number of coalitions, and the number of players taking part in each of them, but independent of the dimension of the various states shared among the parties.

\begin{theorem}\label{hyperghzbell}
Let $G$ be an $N$-player game. Then the maximum bias achievable by players sharing clique-wise entanglement, in which the players are organized in $k$ coalitions of $r$ players each, is at most a constant depending only on $k$ and $r$ times the classical bias. Formally:

$$\betat(G)\,\leq\, 2^{k(3r-5)/2}\, (K_G^{\C})^{k}\, \beta(G).$$
\end{theorem}

Based on a result by Bravyi et al. \cite{bravyi:2006}, we obtain the following corollary, which provides a good example of the applicability of our results.

\begin{corollary}\label{cor:stabilizer}
Let $G$ be a $3$-player game, and let $\ket{\Psi}$ be an arbitrary stabilizer state shared among the three players. Then the following inequality holds:
$$\beta^*_{\ket{\Psi}}(G) \leq 8\, (K_G^{\C})^{4}\,\beta(G).$$
\end{corollary}


\subsection{Application to parallel repetition}

For two-party non-local games the parallel repetition
theorem~\cite{raz:98a, holenstein:2007} states that the bias decreases exponentially in the
number of parallel repetitions of the game.  Closely related to parallel repetition are XOR lemmas.
The $\ell$-fold XOR repetition of an XOR game $G=(f,\pi)$ is again an XOR game defined as 
$G^{\otimes \ell}=(f^{\otimes \ell}, \pi^{\otimes \ell})$ whose game matrix $f^{\otimes l}$ is the $\ell$-fold tensor 
product of the matrix $[f(x,y)]_{x,y}$ and similarly whose distribution is the $\ell$-fold 
tensor product of $[\pi(x,y)]_{x,y}$.  In other words, in this game $\ell$ input pairs 
$(x_i,y_i)_{i=1\ldots l}$ are picked independently with
respect to $\pi$, and all $x_i$ are sent to Alice, $y_i$ to Bob. They should answer
bits $a$ and $b$ respectively such that $a\oplus b = f(x_1,y_1)\oplus\cdots\oplus f(x_\ell,y_\ell)$. 

Cleve et al. \cite{cleve:2008} show that for any XOR game $G$, the game $G^{\otimes \ell}$
has entangled bias $\beta^*(G)^\ell$. Since the classical and quantum biases are within a 
constant factor of each other, this also implies that if $\beta^*(G)<1$, then $\beta(G^{\otimes \ell})$ must go
down exponentially with $\ell$ (although it does not behave as nicely
as the quantum bias with respect to taking XORs).  Cleve et al. further use this XOR lemma to 
give a strong parallel repetition theorem for XOR games with entanglement.  In fact, quite 
generally XOR lemmas imply parallel repetition theorems~\cite{unger:2009}.

Surprisingly, our results (as well as the previous results by \PG~\cite{perezgarcia:2008})
imply that there is no such XOR lemma for classical XOR
games in the $N$-party setting for $N>2$.  This can be seen as
follows.   Suppose that $\beta_S^*(G) =1$ and $\beta(G)
< 1$ for some game $G$.  Then clearly $\beta_S^*(G^{\otimes \ell}) =1$, and so by 
Theorem~\ref{schmidtbound} 
it must be the case that the classical bias of $G^{\otimes \ell}$ is
bounded from below by the constant $(2^{(3N-5)/2}\,K_G^{\C})^{-1}$, which is independent of 
$\ell$.  
Mermin~\cite{mermin:1990} gave an example of such a game, constructing a $3$-party 
XOR game $G$ with the property that $\betad(G)=1$ (in fact the players can always win just by sharing a GHZ state)
 and $\beta(G) = 1/2$.

\subsection{Application to communication complexity}
There is a close connection between the bias of XOR games and 
communication complexity, at least in one direction.  The classical bias of a 
game $G=(f,\pi)$ is equivalent to the {\em discrepancy} of $f$ under the 
distribution $\pi$, up to a constant factor.  The discrepancy method is a 
common way to show lower bounds in communication complexity, and is especially 
valuable in the multiparty model where fewer alternatives are available.  

A simple argument shows that, if a function $f$ has communication 
complexity $c$, then for any distribution $\pi$ the XOR game $G=(f,\pi)$ has bias at least $2^{-c}$.
This is the content of the discrepancy lower bound.
Viewing things in terms of XOR games, however, has certain 
advantages.  The same argument gives that if $f$ has a $c$-bit 
communication protocol where the players share entanglement $\ket{\Psi}$, 
then for any distribution $\pi$, the bias of $G=(f,\pi)$, for players using entanglement $\ket{\Psi}$,
 is at least $2^{-c}$.  As Grothendieck's inequality implies 
that the classical bias and the bias with entanglement are related by a constant 
factor, this gives that the discrepancy method also lower bounds 
the communication complexity with entanglement.  Using teleportation, classical 
communication and entanglement can simulate quantum communication, and so this 
argument already shows that the discrepancy method is a lower bound on quantum 
communication complexity with entanglement.  
This was an open question resolved relatively recently by Linial and 
Shraibman \cite{Linial2007}, who take a point of view dual to the one presented 
here.

These connections readily extend to the multiparty case, for both the 
number-in-the-hand and number-on-the-forehead models of multiparty 
communication complexity.  By showing a 
constant \CQ-gap for clique-wise entanglement, we get that the discrepancy 
method lower bounds multiparty communication complexity where the players 
share GHZ states and EPR-pairs.  Allowing shared EPR-pairs enables
teleportation and so we get the same result allowing quantum communication.
This answers an open question from \cite{lee:2009} who show that the 
discrepancy method is a lower bound on multiparty quantum communication, 
but are not able to handle the case of entanglement.  

Once we can show the discrepancy method is a lower bound on a certain 
model of communication complexity, a by now standard argument leverages this to show 
that the {\em generalized} discrepancy method is as well 
\cite{klauck:2007, razborov:2003, sherstov:2008}.  The generalized 
discrepancy method says that to show a lower bound on a function $f$, it 
suffices to find a function $g$ and probability distribution $\pi$ such that
$g$ has low discrepancy with respect to $\pi$ and $f$ has constant correlation
with $g$ under $\pi$.  In other words, the dual norm of discrepancy can be used to 
lower bound communication complexity.  Disjointness is the canonical example of a function 
which itself has large discrepancy, yet for which the generalized discrepancy 
method can show good lower bounds.

We can summarize our results on communication complexity in the following 
theorem.  
\begin{theorem}\label{thm:cc}
For a sign $N$-tensor $A$, let 
$Q_\epsilon^C(A)$ denote the $N$-party number-on-the-forehead quantum 
communication complexity of $A$ where the players share clique-wise 
entanglement involving at most $k$ subsets of players.  Then
$$
Q_\epsilon^C(A) \ge \frac{1}{2} \max_{B :\ell_1(B)=1} \log 
\left(\frac{\langle A, B \rangle -2\epsilon}{\beta(B)} \right)-O(kN^3).
$$
\end{theorem} 

A proof of this theorem is given in Section~\ref{sec:CC}.

\subsection{Application to hardness of approximation}

Tsirelson's characterization of two-player entangled XOR games 
gives a means to efficiently compute the bias $\beta^*(G)$ to high accuracy
via semidefinite programming.  
It is also known that approximating the \emph{classical} bias of two-player XOR games 
within a sufficiently small constant is NP-Hard~\cite{Hastad:502098}.  Hence the natural relaxation 
that corresponds to allowing the players to share entanglement marks the transition from a hard 
optimization problem to a tractable one, and Alon and Naor \cite{alon:2006} have given 
a constant factor approximation algorithm for the classical bias based on this idea.  
This fact has important consequences when one considers games as 
interactive proof systems, showing the collapse of the class $\oplus$MIP(2) from NEXP to EXP 
(for specific values of completeness and soundness parameters) when the provers are 
allowed to share entanglement~\cite{Cleve:04a}.

As our results show, for multi-player XOR games, generalized Grothendieck inequalities can be 
used 
to bound the \CQ-gap when the provers share specific forms of entanglement, so it is interesting to 
ask whether the quantum bias can again be efficiently approximated. It turns out, however, that the 
situation in this case is quite different. In fact, our results imply the following:

\begin{theorem}\label{thm:inapprox}
For any constant $c>1$, unless P=NP there is no polynomial-time algorithm which approximates the entangled biases $\betad(G)$ or $\betat(G)$ to within a multiplicative factor $c$. Equivalently, for any integer $N$ and any $\eps>0$, unless P=NP there is no polynomial-time algorithm which gives a factor $2-\eps$ approximation to the maximum success probability of an entangled $N$-player game in which the players are restricted to sharing either an arbitrary Schmidt state or any type of clique-wise entanglement.
\end{theorem}

In particular this implies that, while for $2$-player games it is known that $\betad(G)$ can be efficiently approximated using semidefinite programming~\cite{Cleve:04a} (indeed Schmidt states constitute the most general kind of bipartite entanglement), in the case of three or more players, unless P=NP there is no polynomial-time constant-factor approximation algorithm for $\betad(G)$. Note however that our results only hold for the specific types of entanglement that we consider, and it could still very well be the case that, for general entanglement, $\beta^*(G)$ can be computed or approximated in polynomial-time.

The proof of Theorem~\ref{thm:inapprox} follows from a hardness of approximation result for Max-E3-Lin2 due to H\aa stad and Venkatesh~\cite{hastad:2004}, and we give it in Section~\ref{sec:approx}.

\subsection{Application to operator algebras}\label{intro:varopoulos}

A \emph{Banach algebra} $\X = (X,\cdot)$ is a complex Banach space $X$ equipped with a continuous multiplication operation $X\times X\to X:\, (x,y)\mapsto x\cdot y$ which is associative and distributive.
For a Hilbert space $\HS$, let $S_{\infty}$ denote the Banach space of compact operators on $\HS$. The Spectral Theorem characterizes compact operators on a Hilbert space as follows:
\begin{theorem}[Spectral Theorem]
Let $\HS_1,\HS_2$ be Hilbert spaces and let $\HS_1$ have associated inner product $\langle\cdot,\cdot\rangle$. An operator $T:\HS_1\to\HS_2$ is compact if and only if it has a representation of the form
\beqn
T = \sum_i\lambda_i\langle\cdot,e_i\rangle f_i,
\eeqn
where $(e_i)_i, (f_i)_i$ are orthonormal bases for $\HS_1$ and $\HS_2$, respectively, and $\lim_{i \rightarrow \infty} \lambda_i=0$.
\end{theorem}

For ${1\leq p <\infty}$, let $S_p$ denote the space of operators $T\in S_{\infty}$ that satisfy 
 $\Tr(|T|^p)<\infty$, equipped with the Schatten $p$-norm, defined by 
 $\|T\|_p  = \big(\Tr(|T|^p)\big)^{1/p}$. This is called the $p$-Schatten space. 

Varopoulos \cite{varo:1975} asked if for all $1\leq p\leq\infty$, the Banach algebra formed by $S_{p}$ under the Schur 
product (the entry-wise product) is a so-called \emph{Q-algebra}.  A Banach algebra $\X$ with underlying field~$\K$ is a Q-algebra if there exists a constant $C$ such that for any $N\in \N$, and elements $x_1, \ldots, x_N \in \X$ with $\|x_i\|_{\X} \le C$,  any polynomial $q$ in $N$ variables without constant term satisfies the inequality
$$
\|q(x_1,\dots,x_N)\|_{X}\leq \|q\|_{\infty,\K}.
$$

It was proved that for every $p\in[1,4]$, the Banach algebra formed by $S_p$ together with the 
Schur product, is a Q-algebra. The cases $1\leq p\leq 2$ and $2\leq p\leq 4$ were proved by 
P\'{e}rez-Garc\'{i}a~\cite{perez:2006} and {Le Merdy}~\cite{lemerdy:1998}, respectively.

Varopoulos' question is also stated as Open Question 2' in~\cite{perezgarcia:2008}, where a 
connection is made with the \CQ-gap in the case the players share a Schmidt state. 
Using this connection, we complete the answer to the question.
 \begin{theorem}\label{thm:qalgebra}
The Banach algebra formed by $S_{\infty}$ with the Schur product is a $Q$-algebra.
\end{theorem}

We elaborate on this theorem and give a self-contained proof in Section~\ref{app:varopoulos}. 
 It is a simple consequence of the connection already made in~\cite{perezgarcia:2008}, together with our Theorem~\ref{schmidtbound}.

\subsection{Application to Grothendieck-type inequalities}

The constant upper bounds on the \CQ-gaps that we consider are proved by using our expanded connection between \CQ-gaps and known Grothendieck-type inequalities that emerged in Banach-space theory. In some cases, this connection also allows us to prove results in the opposite direction. 

Grothendieck's inequality (cf. Section~\ref{subsec:grothendieck}) may be interpreted as a statement about the inner-product function, which is a linear functional on the tensor product of two Hilbert spaces. Similarly, the extensions of this inequality we consider may be seen as statements about linear functionals on tensor products of multiple Hilbert spaces. Our main technique in proving upper bounds on \CQ-gaps is to relate entangled states to linear functionals for which Grothendieck-type inequalities are known to hold. The linear functionals that correspond to clique-wise entanglement are exactly those which Carne~\cite{carne:1980}  considered and proved to satisfy Grothendieck-type inequalities.

Bravyi et al.~\cite{bravyi:2006} proved that any stabilizer state shared among three parties is equivalent---up to local unitary transformations applied by the parties on their respective systems---to collections of shared GHZ and Bell states (i.e., to 3-partite clique-wise entanglement). Based on our connection, this result and Carne's inequalities imply that the linear functionals that correspond to 3-partite stabilizer states satisfy Grothendieck-type inequalities. We show this more formally in Section~\ref{statestofunctionals}.

\subsection*{Organization of the paper.} We start with some preliminaries in the next section. We then give a proof of our two main results in Sections~\ref{sec:schmidtboundproof} and~\ref{sec:hyperghzbellproof} respectively. In Section~\ref{sec:CC} we elaborate on the connection to communication complexity, while in Section~\ref{sec:approx} we explain the implications of our results to hardness of approximation. Section~\ref{app:varopoulos} is devoted to the proof of our positive answer to the question by Varopoulos, and finally in Section~\ref{statestofunctionals} we explain how our results can lead to new Grothendieck-type inequalities.

\section{Preliminaries}

\subsection{Notation.}

We manipulate finite-dimensional complex Hilbert spaces, usually denoted by $\HS$. Given vectors $x_1,\ldots,x_k \in \HS$, their \emph{generalized inner product} is defined as
$$ \langle x_1,\ldots,x_k\rangle = \sum_{i=1}^d x_1(i)\cdots x_k(i)$$
where $d$ is the dimension of the ambient space $\mathcal{H}$, and $x_l(i)$ refers to the $i$-th coordinate of $x_l$ in the canonical basis. 

If $V$ is a normed vector space, then $\Ball(V)$ denotes the set of all vectors with norm at most $1$. For example, if $\ell_\infty^n$ is the space of all sequences of $n$ real numbers equipped with the supremum norm, then $\Ball(\ell_\infty^n)$ is the set of all sequences of $n$ reals that all have absolute value at most $1$.

An $N$-\emph{tensor} $A \in \K^{d^N}$ is a tensor with entries specified by $N$ 
coordinates $(i_1, \ldots, i_N)$ where $i_j \in [d]$. If $\K=\{-1,1\}$ then $A$ is called a \emph{sign} $N$-tensor.
 For an $N$-tuple of coordinates $I\in [d]^N$, we 
will also write $A[I]$ for the corresponding entry of $A$. If $A, B$ are two tensors of the same 
dimension, we denote their entry-wise product by $A\circ B$: $A\circ B[I] = A[I]\cdot B[I]$ and their inner product by $\langle A,B\rangle = \sum_{I\in[d]^N}A[I]\cdot B[I]$.

\subsection{XOR games}
An $N$-player \emph{XOR game} can be described as follows. Let $n$ be a positive 
integer, $\pi$ a probability distribution on $[n]^N$, and $A:[n]^N\to \pmset{}$ a sign $N$-tensor. At the start of the game 
$G=(A,\pi)$, a verifier picks an $N$-tuple $(i_1,\dots,i_N)$ of questions according to $\pi$ and 
distributes these among the players. The players, who may agree on a strategy beforehand but are 
not allowed to communicate after receiving their questions, must answer the verifier with bits 
$x_1(i_1),\dots,x_N(i_N)\in\pmset{}$. They win the game if the product of their answers equals 
$A[i_1,\dots,i_N]$. The player's success is given in terms of the \emph{bias} of the game:

\begin{definition} 
Let $M$ be a real or complex $N$-tensor.  Define
\beqn
\beta(M):=\max_{x_1,\ldots,x_N:[n]\to\pmset{}}\Big|\sum_{i_1,\ldots,i_N=1}^{n} M[i_1,\dots,i_N]\, x_1(i_1)\cdots x_N(i_N)\Big|
\eeqn
The \emph{classical bias} of a game $G=(A,\pi)$ is $\beta(A \circ \pi)$.
\end{definition}

With some abuse of notation, we often denote the classical bias by $\beta(G)$.
With this definition, the maximum probability under $\pi$ with which the players can win 
the game is simply $1/2+\beta(G)/2$.

We also consider XOR games where the players share entanglement. 
In this setting, the players are allowed to have quantum systems described by Hilbert spaces 
$\HS_1,\dots,\HS_N$, respectively. Before the game begins, the players put the overall system in 
an entangled state $\psiket\in\HS_1\otimes\cdots\otimes\HS_N$ and decide on sequences of $
\pmset{}$-valued measurement observables $X_1(i_1),\dots,X_N(i_N)$ on their respective Hilbert 
spaces. Upon receiving their questions $i_1,\dots,i_N$, the players answer with the 
measurement outcomes $x_1(i_1),\dots,x_N(i_N)$ of their respective observables, when 
performed on their share of $\psiket$. The expectation of the product of their answers 
is exactly $\psibra X_1(i_1)\otimes\cdots\otimes X_N(i_N)\psiket$. We can now define an 
entangled version of the bias.

\begin{definition}
Let $M$ be a real or complex $N$-tensor. Define
\beqn
\beta^*(M):=\max_{\psiket,X_1,\dots,X_N}\Big|\sum_{i_1,\ldots,i_N=1}^{n}M[i_1,\dots,i_N]\,
\psibra X_1(i_1)\otimes\cdots\otimes X_N(i_N)\psiket\Big|.
\eeqn
The \emph{entangled bias} of a game $G=(A,\pi)$ is $\beta^*(A \circ \pi)$.

\end{definition}

\subsection{Grothendieck inequalities.}\label{subsec:grothendieck}
Let $\K$ denote either $\R$ 
or $\C$. Then the (real or complex) Grothendieck constant \emph{of order $d$}, denoted by $K_G^{\K}(d)$, is the smallest positive constant such that for every $n_1,n_2\in\N$, for every matrix $M = (M_{ij})
\in \K^{n_1\times n_2}$ and vectors $u_1,\dots,u_{n_1},v_1,\dots,v_{n_2} \in\Ball(\K^d)$, 
the following inequality holds:
\beq\label{grothineq}
\sum_{i=1}^{n_1} \sum_{j=1}^{n_2} M_{ij}\, \langle u_i, v_j \rangle \leq K_G^{\K}(d)
\mathop{\max_{x:[n_1]\to\Ball(\K)}}_{y:[n_2]\to\Ball(\K)}\Big|\sum_{i=1}^{n_1}\sum_{j=1}^{n_2}M_{ij} x(i)y(j)\Big|.
\eeq
The real and complex Grothendieck constants are defined as $K_G^{\K} := \sup_dK_G^{\K}(d)$.
The exact values of $K_G^{\R}$ and $K_G^{\C}$ are still unknown, but they have been shown to be bounded as follows: $1.6770\lesssim K_G^{\R}\lesssim 1.7822$ and $1.3380\lesssim K_G^{\C}\lesssim 1.4049$.

Tsirelson \cite{tsirelson93:_some} gave a very elegant characterization of the bias of two-player 
XOR games with entanglement.  Namely, for an $n$-by-$n$ sign matrix $A$ and probability 
distribution $\pi$ he showed
$$
\beta^*(A \circ \pi)=\max_{u_i,v_j\in\Ball(\R^n)}\sum_{i,j=1}^{n} A[i,j]\pi(i,j)\, \langle u_i, v_j \rangle.
$$
Combined with Grothendieck's inequality (\ref{grothineq}), this yields 
$\beta^*(A \circ \pi)\leq K_G^{\R}\,\beta(A \circ \pi)$, for any sign matrix $A$ and probability 
distribution $\pi$. This shows that for two-player XOR games the \CQ-gap is bounded by a constant. 


To facilitate our discussion on extensions of Grothendieck's inequality, we will use the following notation.

\begin{definition}\label{phinorm}
Let $\K,\K' = \R$ or $\C$. For an $N$-tensor $A:[n]^{N}\to\K$, define its norm
\beqn
\|A\|_{\infty,\K'}:= \max_{\phi_1,\dots,\phi_N\in\Ball(\ell_{\infty}^n)}\Big|\sum_{i_1,\dots,i_N=1}^nA[i_1,\dots,i_N]\phi_1(i_1)\cdots \phi_N(i_N)\Big|,
\eeqn
where the underlying scalar field for $\ell_{\infty}^n$ is $\K'$. Also, define
\[
\gamma^*(A)=\sup_d \sup_{\phi_1,\dots,\phi_N:[n]\to\Ball(\C^d)} 
\Big|\sum_{i_1,\dots,i_N=1}^nA[i_1,\dots,i_N]\langle \phi_1(i_1),\ldots, \phi_N(i_N)\rangle \Big|.
\]
\end{definition}

If $\K'=\R$ and $A$ is a tensor that was constructed from a game $G$, then $\|A\|_{\infty,\R}$ is exactly the 
classical bias $\beta(G)$ of the game.

\subsection{Multipartite entanglement}\label{sec:entanglement}
The following definition will be useful in studying the different biases achievable by players who are restricted to sharing a specific type of entanglement. 

\begin{definition}
Let $A$ be a $N$-dimensional tensor, and $\ket{\Psi}\in \HS^{\otimes N}$ a fixed entangled state 
shared by $N$ players. Then the \emph{bias restricted to $\ket{\Psi}$}, denoted 
$\beta^*_{\ket{\Psi}}(A)$, is defined as
$$\beta^*_{\ket{\Psi}}(A)\,=\, \max_{M_1,\ldots,M_N} \Big|\sum_{i_1,\ldots,i_N} A[i_1,\ldots, i_N] \bra{\Psi} 
M_1(i_1)\otimes\cdots\otimes M_k(i_N)\ket{\Psi}\Big|$$
where the maximum is taken over all sets of $\pmset{}$-valued observables $M_\ell(i_{\ell})$ on $\HS$.  
For a game $G=(A, \pi)$, we will also write $\beta^*_{\ket{\Psi}}(G)$ for 
$\beta^*_{\ket{\Psi}}(A \circ \pi)$.
\end{definition}

The following setups are the ones that we will encounter most frequently, and for each we 
introduce a special notation for the bias. For the case of GHZ states $\ket{\Psi} = d^{-1/2}\sum_{i=1}^d\ket{i}_1\cdots\ket{i}_N$ (of arbitrary dimension $d$) we will denote the maximum bias by $\betag(G)$, while for Schmidt states $\ket
{\Psi} = \sum_{i=1}^d \alpha_i \ket{i}_1\cdots\ket{i}_N$ (with arbitrary dimension $d$ and coefficients $\alpha_i$) we will use the notation $\betad(G)$.
Finally, clique-wise entanglement is any type of entanglement that can be obtained by grouping the $N$ players into $k$ coalitions of $r$ players each (a given player can take part in any number of coalitions), and allowing the members of each of the coalitions to share a GHZ state of arbitrary dimension. In that case, we denote the maximal bias by $\betat(G)$. This may depend on the parameters $k$ and $r$, which are kept implicit so as not to overload the notation, but will always be clear in context.

We have the following obvious relationships between the biases: 
$$\beta(G)\, \leq\, \betag(G) \,\leq \,\betad(G)\, \leq \,\beta^*(G) \qquad\text{and}\qquad 
\beta(G)\,\leq\,\betag(G)\,\leq\,\betat(G)\, \leq \,\beta^*(G)$$

\section{Proof of Theorem~\ref{schmidtbound}}\label{sec:schmidtboundproof}

In this section we give a proof of Theorem~\ref{schmidtbound}. 
First, in Section~\ref{sec:ghzbound} we analyze the maximum bias $\betag(G)$ achievable by 
strategies that are limited to sharing a GHZ state, and show that it is upper-bounded by the 
maximum of an expression involving generalized inner products of $n$-dimensional vectors.  
This section can be seen as a warm-up to give the reader the general idea of our proofs.
Then, in Section~\ref{sec:tonge}, we adapt a theorem by Tonge~\cite{tonge:1978} to show that the 
optimum of this problem is upper-bounded by a constant times the classical bias of the original 
game, thus showing a constant \CQ-gap for the case of GHZ states. Finally, in Section~\ref
{sec:diagonal}, we extend our proof to cover the case where the players are allowed to share a 
Schmidt state.

\subsection{Strategies with GHZ states.}\label{sec:ghzbound}
We prove the following lemma:
\begin{lemma}\label{lem:ghzbound}
Let $G=(A,\pi)$ be an $N$-player game. Assume that the players are restricted to sharing a state 
of the form $\ket{\Psi} = \frac{1}{\sqrt{d}}\sum_{i=1}^d \ket{i}^{\otimes N}$. Then the maximum bias 
the players can achieve is upper-bounded as follows:
$$ \betag(G)\,\leq\, \gamma^*(A\circ \pi).$$
\end{lemma}

\begin{proof}
Fix any strategy of the players using entanglement $\ket{\Psi}$, and let $M_\ell(i)$ be the 
$\pmset{}$ valued observable of player $\ell$ on question $i$ in that strategy. Let $B=A\circ\pi$.
The players' bias is given by 
\begin{align*}
\betag(G)&=
\Big|\sum_{i_1,\ldots,i_N} B[i_1,\ldots,i_N]\bra{\Psi} M_1(i_1) \otimes 
\cdots \otimes M_N(i_N)  \ket{\Psi} \Big|\\
& =\frac{1}{d} \Big|\sum_{i_1,\ldots,i_N}  B[i_1,\ldots,i_N] \sum_{i,j=1}^d 
\bra{i}M_1(i_1)\ket{j} \cdots \bra{i}M_N(i_N)\ket{j}\Big|\\ 
& \leq \frac{1}{d} \sum_{j=1}^d\Big| \sum_{i_1,\ldots,i_N}  B[i_1,\ldots,i_N]
\sum_{i=1}^d \bra{i}M_1(i_1)\ket{j} \cdots \bra{i}M_N(i_N)\ket{j} \Big|\\
&\le \frac{1}{d} \sum_{j=1}^d \gamma^*(B) =\gamma^*(B).
\end{align*}
The inequality holds as the inner sum on the third line is a generalized inner product 
of the $j^{th}$ columns of the $M_k(i_k)$'s, which are unit vectors since these matrices
are unitary.
\end{proof}

\subsection{Tonge's theorem}\label{sec:tonge}
In this section we introduce a slight generalization of a result originally due to Blei~\cite{blei:1979} 
and later improved and simplified by Tonge~\cite{tonge:1978}.  This theorem allows us
 to relate the maximization over generalized inner products from Lemma~\ref{lem:ghzbound} to the 
 classical bias of the game and thus prove Theorem~\ref{schmidtbound} for the case of GHZ states.

\begin{theorem}\label{kgrothbound}
Let $n,N\geq 2$ and $d$ be positive integers, and $\HS$ be a $d$-dimensional complex Hilbert 
space. Then, for every $N$-tensor $B:[n]^{N}\to \R$ and 
$f_1,\dots,f_N:[n]\to \Ball
(\HS)$, the following inequality holds:
\beq\label{tongeineq}
\Big|\sum_{i_1,\dots,i_N=1}^n B[i_1,\dots,i_N]\langle f_1(i_1),\dots,f_N(i_N)\rangle\Big| \leq 
2^{(3N-5)/2}K_G^{\C}\|B\|_{\infty,\R}
\eeq
\end{theorem}

This theorem is proved in Appendix~\ref{apptonge}. Since $\|A\circ \pi\|_{\infty,\R}$ is exactly the classical bias of $G=(A,\pi)$, combining Lemma~\ref{lem:ghzbound} with Theorem~\ref{kgrothbound} leads to the bound $\betag(G) \leq 2^{(3N-5)/2}\,K_G^{\C}\, \beta(G)$, thus proving Theorem~\ref{schmidtbound} for the special case of GHZ states.

\subsection{Extension to Schmidt states.}\label{sec:diagonal}
We extend the result of Section~\ref{sec:ghzbound} to the case of Schmidt states, thus proving Theorem~\ref{schmidtbound} in full generality. For this, analoguously to Lemma~\ref{lem:ghzbound}, it is suffient to show that, if $\ket{\Psi}=\sum_{i=1}^d \alpha_i \ket{i}^{\otimes N}$ is a Schmidt state and $B$ an $N$-tensor, then 
\[
\beta^*_{\ket{\Psi}}(B)\, \le\,  \gamma^*(B)
\]
The theorem will follow by setting $B=A\circ\pi$ for a $N$-player game $G=(A,\pi)$, and applying Theorem~\ref{kgrothbound}.

\begin{proof}[ of Theorem~\ref{schmidtbound}]
Let $B$ be an $N$-tensor, and $\{M_j(x_j)\}_{j,x_j}$ a choice of $\pm 1$ valued observables which achieve the maximal bias $\beta^*_{\ket{\Psi}}(B)$. By absorbing any complex phases into 
the strategy of player one, we can assume that all coefficients $\alpha_i$ are positive reals. The 
following claim shows that $\ket{\Psi}$ can be  expressed as a weighted sum of GHZ-type states.

\begin{claim}\label{claim:generalizedghz}
Let $\ket{\Psi} = \sum_{i=1}^d \alpha_i \ket{i}^{\otimes N}$ be a (normalized) state such that the $\alpha_i$ are positive reals. Then there exists positive reals $\beta_1,\ldots,\beta_d$ such that $\ket{\Psi}=\sum_{\ell=1}^d \beta_\ell \ket{\phi_\ell}$, where  $\ket{\phi_\ell} = \sum_{i=1}^\ell \ket{i}^{\otimes N}$ for $\ell=1, \ldots, d$ is a ``partial'' (un-normalized) GHZ state. Moreover, the $\beta_i$ satisfy the following equation:
\begin{equation}
\label{norm}
\sum_{i,j=1}^d \beta_i \beta_j \cdot  \min\{i,j\} =1
\end{equation}
\end{claim}

\begin{proof}
Renaming the basis vectors as 
necessary, we can assume that $\alpha_d \le \cdots \le \alpha_1$.   Let $\beta_d=\alpha_d$ 
and $\beta_i=\alpha_i-\alpha_{i+1}$ for $i=1,\ldots, d-1$.  Then we have
$$
\ket{\Psi}=\sum_{i=1}^d \beta_i \ket{\phi_i}.
$$
Moreover, Eq.~(\ref{norm}) is immediate from the fact that $|\braket{\Psi}{\Psi}|=1$ and $\braket{\phi_i}{\phi_j} = \min\{i,j\}$ (recall that $
\ket{\phi_\ell}$ itself was not normalized).
\end{proof}

This reformulation of $\ket{\Psi}$ reduces the task of showing an upper bound 
on $\beta^*_{\ket{\Psi}}(B)$ to a form similar to what we had before.  Namely,
$$
\beta^*_{\ket{\Psi}}(B)=\sum_{i,j} \beta_i \beta_j \sum_{x_1, \ldots, x_N}B[x_1,\dots,x_N] \bra{\phi_i} 
M_1(x_1) \otimes \cdots \otimes M_N(x_N) \ket{\phi_j}.
$$
For fixed  $i,j$, each term of the sum involves unnormalized ``partial'' GHZ states, which can be handled in 
the same fashion as Lemma~\ref{lem:ghzbound}.
\begin{claim}\label{claim:partialghz}
Let $B$ be an $N$-tensor and $\{M_k(x_k)\}$ be $\pm 1$ valued observables.  Then
$$
\sum_{x_1, \ldots, x_N} B[x_1, \ldots, x_N] 
\bra{\phi_i} M_1(x_1) \otimes \cdots \otimes M_N(x_N) \ket{\phi_j} \le 
\min\{i,j\} \gamma^*(B).
$$
\end{claim}

\begin{proof}
\begin{align*}
\sum_{x_1, \ldots, x_N} B[x_1, \ldots, x_N] 
\bra{\phi_i} M_1(x_1) \otimes &\cdots \otimes M_N(x_N) \ket{\phi_j} = \\
&\sum_{x_1, \ldots, x_N} B[x_1, \ldots, x_N] \sum_{s=1}^i \sum_{t=1}^j 
\bra{s}M_1(x_1) \ket{t} \cdots \bra{s}M_N(x_N)\ket{t}
\end{align*}
We will order the double sum over $s,t$ depending on whether $i$ or $j$ is smaller---we 
want the outer sum to be over the smaller one.  Suppose that $i \le j$.  The other case is 
completely analogous.  Then
\begin{align*}
\sum_{x_1, \ldots, x_N} B[x_1, \ldots, x_N] \sum_{s=1}^i \sum_{t=1}^j 
\bra{s}M_1(x_1) \ket{t} &\cdots \bra{s}M_N(x_N)\ket{t}= \\
&\sum_{s=1}^i \sum_{x_1, \ldots, x_N} B[x_1, \ldots, x_N] \sum_{t=1}^j 
\bra{s}M_1(x_1) \ket{t} \cdots \bra{s}M_N(x_N)\ket{t}
\end{align*}
For each fixed $s$, the inner sum is now a generalized inner product of the first $j$ entries of 
the $s^{th}$ row of the $M_k(x_k)$'s.  Since these have norm at 
most one, we have 
$$
\sum_{x_1, \ldots, x_N} B[x_1, \ldots, x_N] 
\bra{\phi_i} M_1(x_1) \otimes \cdots \otimes M_N(x_N) \ket{\phi_j} \le 
\min\{i,j\} \gamma^*(B).
$$
\end{proof}

We can now finish the proof.
\begin{align*}
\beta^*_{\ket{\Psi}}(B)&=\sum_{i,j} \beta_i \beta_j \sum_{x_1, \ldots, x_N} \bra{\phi_i} 
M_1(x_1) \otimes \cdots \otimes M_N(x_N) \ket{\phi_j} \\
&\le \sum_{i,j} \beta_i \beta_j \min\{i,j\}  \gamma^*(B) \\
&=\gamma^*(B).
\end{align*}
The first inequality follows from Claim~\ref{claim:partialghz} and the second by
Claim~\ref{claim:generalizedghz}.
\end{proof}


\section{Proof of Theorem~\ref{hyperghzbell}}\label{sec:hyperghzbellproof}

The proof of Theorem~\ref{hyperghzbell} is based on a result by Carne \cite{carne:1980}, which 
essentially shows how Grothendieck-type inequalities can be composed in order to prove new 
inequalities of the same type. This will let us prove bounds on the entangled bias when the 
players are allowed to share any combination of EPR pairs and GHZ states. We first explain 
Carne's theorem in Section~\ref{sec:carne}, for which we give a self-contained proof in Appendix~\ref{app:carne}. 
We explain how it is applied to prove Theorem~\ref{hyperghzbell} in Section~\ref
{sec:applyingcarne}. We will end this section with a proof of Corollary~\ref{cor:stabilizer}.

\subsection{Carne's theorem}\label{sec:carne}

Carne \cite{carne:1980} showed that inequalities such as the one by Blei and Tonge (Theorem~\ref{kgrothbound}), could be composed in order to prove more general inequalities. His result also shows how Tonge's inequality can be re-derived as a consequence of the original Grothendieck inequality, though Tonge's version of inequality~\eqref{tongeineq} is tighter. 

To describe Carne's theorem, consider a hypergraph $H = (V,E)$, and associate with each edge $e\in E$ and vertex $x\in e$ a complex Hilbert space $\HS(x,e)$. Define $\HS_x := \bigotimes_{e\in E(x)}\HS(x,e)$. 
Assume that, with each edge $e\in E$ is associated a multi-linear continuous functional $\psi_e:\bigotimes_{x\in e}\HS(x,e)\to\C$ that satisfies a Grothendieck-type inequality, i.e. for every $|e|$-tensor $D:[n]^{e}\to\K = \R$ or $\C$ and set of functions $f_x:[n]\to\Ball(\HS_x)$, for each $x\in e$, the inequality
\beq\label{carneineq}
\sum_{K\in[n]^{|e|}}D[K]\psi_e\Big(\bigotimes_{x\in e}f_x(i_x)\Big) \leq C_e^{\K'}\, \|D\|_{\infty,\K'}
\eeq
holds for a field $\K'=\R$ or $\C$, and some constant $C_e^{\K'}$ independent of the tensor $D$. Carne's theorem then states that the natural combination of the linear functionals $\psi_e$ in a general multi-linear functional $\Phi$ defined over the whole Hilbert space $\HS = \otimes_{x\in V} \HS(x)$ also satisfies a Grothendieck-type inequality, with underlying constant the product of the $C_e^{\K'}$. Note that, since a vertex $x$ can be part of many edges, there can be many functionals $\psi_e$ which act on the same space $\HS(x)$. This is what makes Carne's theorem non-trivial. To state it we need to define the linear \emph{re-arranging map} $\sigma$ as
\beqn
\sigma: \bigotimes_{x\in V}\Big(\bigotimes_{e\in E(x)}\HS(x,e)\Big) \to \bigotimes_{e\in E}\Big(\bigotimes_{x\in e}\HS(x,e)\Big),
\eeqn
which simply permutes the factors of a vector $v\in\bigotimes_{x\in V}\HS_x$. 

\begin{theorem}[Slight extension of Carne 1980]\label{carnethm}
Let $\K,\K'\in\{\R,\C\}$. The  linear functional defined by $\Phi := \Big(\bigotimes_{e\in E}\psi_e\Big)\circ \sigma$ satisfies  that for every $|V|$-tensor $A:[n]^V\to\K$ and set of functions $f_x:[n]\to\Ball\big(\HS_x\big)$, for $x\in V$, the following inequality holds:
\beq\label{carnegenineq}
\sum_{I\in[n]^V}A[I]\cdot \Phi\Big(\bigotimes_{x\in V}f_x(i_x)\Big) \leq \Big(\prod_{e\in E}C_e^{\K'}\Big)\|A\|_{\infty,\K'}.
\eeq
\noindent where the constants $C_e^{\K'}$ are such that~\eqref{carneineq} holds. 
\end{theorem}

In particular, if $\K=\R$ and each $\psi_e$ is the generalized inner product function on $\bigotimes_{x\in e}\HS(x,e)$, then it follows from Theorem~\ref{kgrothbound} that the constant in \eqref{carnegenineq} is upper bounded by $\Big(\prod_{e\in E}2^{(3|e|-5)/2}\Big)\, (K_G^{\C})^{|E|}$.

\subsection{Bounding the bias achievable by strategies with clique-wise entanglement}\label{sec:applyingcarne}
Consider an $N$-player game $G=(A,\pi)$. Let the players be organized in $k$ coalitions of $r$ players each\footnote{The organization of these coalitions is independent of the game itself; rather it is used to define the structure of the entanglement that is shared between the players.}, where a given player can take part in any number of coalitions. Each coalition of players is allowed to share a GHZ state between its members.

To model this setup, associate a hypergraph $H=(V,E)$ to the coalition structure, with $V=[N]$ and there is a hyperedge for every coalition. For every edge $e$ we introduce a Hilbert space $\HS(e) = \otimes_{x\in e} \HS(x,e)$, where $\HS(x,e)$ is the local space of player $x$ corresponding to edge $e$. The state of the players in this space is initialized in a GHZ state $\ket{\Psi_e} = d^{-1/2} \sum_{i=1}^d \ket{i}^{\otimes |e|}$. The global entangled state shared by the players at the start of the game is then 
\begin{align}\label{eq:psitilde}
\ket{\tilde{\Psi}} &= \otimes_{e\in E} \ket{\Psi_e} \, \in \, \bigotimes_{e\in E} \Big( \bigotimes_{x\in e} \HS(x,e) \Big)
\end{align}
 Finally, each player $x$ has observables $M_x(i)$ corresponding to question $i$. These act on player $x$'s local space $\HS(x) = \otimes_{e\ni x} \HS(x,e)$.

Theorem~\ref{hyperghzbell} states that the maximum bias achievable by a strategy of the form that we have just described is at most a constant times the classical bias of the game. In order to prove it, we first relate the bias achieved by any strategy to an expression similar to the one appearing on the left-hand side of~(\ref{carnegenineq}) in Carne's theorem, where $\psi_e$ will be the linear functional that is associated with the GHZ state, i.e. the generalized inner product function. Applying Theorem~\ref{carnethm} will conclude the argument.

\begin{proof}[ of Theorem~\ref{hyperghzbell}]
Fix observables $M_x$ and an entangled state $\ket{\Psi} \in \otimes_{x\in V} \HS_x$ of the form described above. Note that $\ket{\Psi} = \sigma^{-1}(\ket{\tilde{\Psi}})$, where $\ket{\tilde{\Psi}}$ is described in Eq.~(\ref{eq:psitilde}). This is because we need to re-arrange the terms in the definition of $\ket{\tilde{\Psi}}$ to correspond to the decomposition of space $\otimes_{x\in V} \HS_x$.

 We begin by expanding $\bra{\Psi}\bigotimes_{x\in V}M_x\ket{\Psi}$, with the goal of relating it to the map $\Phi$ of Theorem~\ref{carnethm}. Let $[d]^E$ denote the set of $|E|$-tuples $(j_e)_{e\in E}$. We have
\beqrn
\ket{\Psi} &=& \sigma^{-1}\left(\frac{1}{\sqrt{d^{|E|}}}\:\bigotimes_{e\in E}\Big(\sum_{e_i=1}^d\bigotimes_{x\in e}\ket{j_{e_i}}\Big)\right)\\
&=& \frac{1}{\sqrt{d^{|E|}}}\sum_{J\in[d]^E}\:\bigotimes_{x\in V}\ket{J_{|E(x)}}
\eeqrn
where $J_{|E(x)}$ denotes the restriction of the tuple $J\in [d]^E$ to those edges that contain the vertex $x$. Since observables are Hermitian, the expected value $\bra{\Psi}\bigotimes_{x\in V}M_x\ket{\Psi}$ is given by
\beqr
\bra{\Psi}\bigotimes_{x\in V}M_x\ket{\Psi} 
&=& \frac{1}{2\cdot d^{|E|}}\sum_{J',J\in[d]^E}\:\left(\prod_{x\in V}\bra{J'_{|E(x)}}M_x\ket{J_{|E(x)}} + \prod_{x\in V}\bra{J_{|E(x)}}M_x\ket{J'_{|E(x)}}\right)\nonumber\\
&=& \frac{1}{2\cdot d^{|E|}}\sum_{J',J\in[d]^E}\:\left(\prod_{x\in V}\bra{J'_{|E(x)}}M_x\ket{J_{|E(x)}} + \prod_{x'\in V}\bra{J'_{|E(x')}}M_{x'}^*\ket{J_{|E(x')}}\right)\nonumber\\
&=& \frac{1}{d^{|E|}}\sum_{J'\in[d]^E}\left(\sum_{J\in[d]^E}\:\Re\big(\prod_{x\in V}\bra{J'_{|E(x)}}M_x\ket{J_{|E(x)}}\Big)\right)\nonumber\\
&=& \frac{1}{d^{|E|}}\sum_{J'\in[d]^E}\Re\left(\sum_{J\in[d]^E}\:\prod_{x\in V}\big[M_x\big]_{J_{|E(x)},\, J'_{|E(x)}}\right),\label{hyperexpect}
\eeqr
where the subscript $(J_{|E(x)},J'_{|E(x)})$ indicates a row-column pair of the matrix $M_x$. Note that, since the expression on the left-hand side is real, the one on the right is too, and we can safely ignore the $\Re$ symbol on the right. Since the $M_x$ are unitary matrices, their columns are unit vectors. This implies that there exists unit vectors $v_x\in\bigotimes_{e\in E(x)}\HS(x,e)$ (depending on $J'$) such that the expression between the brackets in equation~\eqref{hyperexpect} is of the form
\beqn
\sum_{J\in [d]^E}\:\prod_{x\in V}v_x(J_{|E(x)})
\eeqn
where as usual $v_x(J_{|E(x)})$ denotes the restriction of the vector $v_x$ to those indices in $J_{|E(x)}$.

\begin{claim}\label{claim:phimap}
For $\Phi:= \Big(\bigotimes_{e\in E}\psi_e\Big)\circ\sigma$ with $\psi_e$ the generalized inner product function on $\bigotimes_{x\in e}\HS(x,e)$, we have
\beqn
\sum_{J\in [d]^E}\:\prod_{x\in V}v_x(J_{|E(x)}) = \Phi\Big(\bigotimes_{x\in V}v_x\Big).
\eeqn
\end{claim}

\begin{proof} Since $\Phi$ is linear, it suffices to prove the claim for vectors of the form $v_x = \bigotimes_{e\in E(x)}v_{x,e}$, where each $v_{x,e}\in\HS(x,e)$. In this case, we have 
\beqrn
\Big(\bigotimes_{e\in E}\psi_e\Big)\circ \sigma\Big(\bigotimes_{x\in V}\big(\bigotimes_{e\in E(x)}v_{x,e}\big)\Big) &=& \bigotimes_{e\in E}\Big(\psi_e\big(\bigotimes_{x\in e}v_{x,e}\big)\Big)\\
&=& \prod_{e\in E}\Big(\sum_{j_e=1}^d\Big(\prod_{x\in e}v_{x,e}(j_e)\Big)\Big)\\
&=& \sum_{J\in[d]^E}\prod_{e\in E}\Big(\prod_{x\in e}v_{x,e}(j_e)\Big)\\
&=& \sum_{J\in[d]^E}\:\prod_{x\in V}\Big(\prod_{e\in E(x)}v_{x,e}(j_e)\Big),
\eeqrn
where the last product is  $\prod_{e\in E(x)}v_{x,e}(j_e) = v_x(J_{|E(x)})$.
\end{proof}

Let $M_x(i)$ be the observable used by player $x$ on question $i$, so that the bias achieved by this strategy in the game $G=(A,\pi)$ is
$$\Big|\sum_{I\in [n]^V} B[I]\, \bra{\Psi} \bigotimes_{x\in V} M_x(i_x) \ket{\Psi} \Big|$$
where $B = A\circ\pi$. 
We can bound this expression by
\begin{multline}\label{realcarne}
\Big|\sum_{I\in [n]^V}B[I]\left( \frac{1}{d^{|E|}}\sum_{J'\in[d]^E}\sum_{J\in[d]^E}\:\prod_{x\in V}\big[M_x(i_x)\big]_{J_{|E(x)},\, J'_{|E(x)}}\right)\Big|\\
 \leq  \frac{1}{d^{|E|}}\sum_{J'\in[d]^E}\Big|\sum_{I\in [n]^V}B[I]\cdot\sum_{J\in[d]^E}\:\prod_{x\in V}\big[M_x(i_x)\big]_{J_{|E(x)},\, J'_{|E(x)}}\Big|\\
 \leq \max_{J'\in[d]^E}\Big|\sum_{I\in [n]^V}B[I]\cdot\sum_{J\in[d]^E}\:\prod_{x\in V}\big[M_x(i_x)\big]_{J_{|E(x)},\, J'_{|E(x)}}\Big|\\
 \leq \max_{f_x:[n]\to\Ball(\HS_x):\, x\in V}\Big|\sum_{I\in[n]^V}B[I]\cdot \Phi\Big(\bigotimes_{x\in V}f_x(i_x)\Big)\Big|,
\end{multline}
where the first equality is~\eqref{hyperexpect}, and the last inequality follows from Claim~\ref{claim:phimap}. The result then follows directly from Theorem~\ref{carnethm} combined with the bound in Theorem~\ref{kgrothbound}, giving the last part of the theorem. 
\end{proof}

We end this section with a proof of  Corollary~\ref{cor:stabilizer}.

\begin{proof}[ of Corollary~\ref{cor:stabilizer}]
Theorem 5 in~\cite{bravyi:2006} states that, if $\ket{\Psi}$ is any stabilizer state shared in an arbitrary way among three parties, then $\ket{\Psi}$ is local-unitarily equivalent to a number of EPR pairs shared between each of the three pairs of players, together with a GHZ state shared in common. They even give the number of such states, based on the structure of the initial stabilizer state. It now suffices to consider the hypergraph $G$ with vertex set $V=\{1,2,3\}$, and edge set $E=\{\{1,2\},\{2,3\},\{1,3\},\{1,2,3\}\}$. In the notation of Theorem~\ref{hyperghzbell}, this hypergraph has $k=4$ and $r\leq 3$, which gives the bound $2^8 (K_G^{\C})^4$. However, a careful examination of the proof of Theorem~\ref{hyperghzbell} easily reveals that the inequality holds with the smaller constant~$8 (K_G^{\C})^4$.
\end{proof}

\section{Application to communication complexity}\label{sec:CC}
In this section we give a proof of Theorem~\ref{thm:cc}, showing a lower bound on quantum 
multiparty communication complexity for clique-wise entanglement. 

For a sign $N$-tensor $A$, let $R_\epsilon(A)$ denote the multiparty randomized communication 
complexity of $A$ with error at most $\epsilon$, and let $R_\epsilon^{\psiket}(A)$ denote the 
minimal cost of an $\epsilon$ bounded-error protocol where the players share a state $\psiket$, 
and communicate classical bits.  Finally, let $Q_\epsilon^{\psiket}(A)$ denote the minimal 
cost of a protocol in the strongest model we will consider---where the players share entanglement
$\psiket$ and use quantum communication.  We refer the reader to \cite{lee:2009} for a 
description of the multiparty quantum model of communication.  The reader should think of 
all these measures in the number-in-the-hand (NIH) model of communication; at the end we will 
explain why the results also hold in the number-on-the-forehead (NOF) model.

The generalized discrepancy method is a very useful lower bound method for randomized 
communication complexity and still essentially the only lower bound method available in the 
NOF model of multiparty complexity.  
It was developed over a sequence of works for the two-party and multiparty models
\cite{klauck:2007, razborov:2003, sherstov:2008,lee-shraibman:2009, chattopadhyay:2008}.
\begin{theorem}
Let $A$ be a sign $N$-tensor.  Then
$$
2^{R_\epsilon(A)} \ge \max_{B,\pi} \frac{\langle A, B \circ \pi \rangle -2\epsilon}{\beta(B \circ \pi)}
$$
where the maximization is over all sign $N$-tensors $B$ and probability distributions $\pi$. 
\label{thm:gen_disc}
\end{theorem}

We start by proving a result exactly analogous to Theorem~\ref{thm:gen_disc} for protocols with 
entanglement.
\begin{proposition}
Let $A$ be a sign $N$-tensor.  Then for any state $\psiket$
$$
2^{R_\epsilon^{\psiket}(A)} \ge 
\max_{B,\pi} \frac{\langle A, B \circ \pi \rangle -2\epsilon}{\beta_{\psiket}^*(B \circ \pi)} 
$$
where the maximization is over all sign $N$-tensors $B$ and probability distributions $\pi$. 
\label{thm:entangledCC}
\end{proposition}

\begin{proof}
Consider a communication protocol with entanglement $\psiket$ for $A$ of minimal cost $c$ 
and error at most $\epsilon$.  Let $R$ be the $N$-tensor such that $R[x_1, \ldots, x_N]$ is 
the expectation of the output of this protocol on input $(x_1, \ldots, x_N)$.  
By assumption of the correctness of the protocol, if $A[x_1, \ldots, x_N]=1$ then
$1-2\eps \le R[x_1, \ldots, x_N] \le 1$ and if $A[x_1, \ldots, x_N]=-1$ then 
$-1 \le R[x_1, \ldots, x_N] \le -1+2\eps$.

Fix a probability distribution $\pi$ and let $B$ be an arbitrary sign
tensor of the same dimensions as $A$.  We will see how the communication
protocol for $A$ can be used to design a XOR protocol for $B$.
The bias of this protocol will be related to the amount of communication $c$ and the
correlation $\langle A,B \circ \pi \rangle$.

The strategy in the XOR game is as follows.  We may assume that the players have 
access to a shared random string $r$.  A convexity argument shows that 
shared randomness cannot increase the bias.  
On input $(x_1, \ldots, x_k)$ the players look at the shared random string $r$ of 
length $c$.  The players interpret $r$ as a ``guess'' for the transcript of 
the communication protocol on input $(x_1, \ldots, x_N)$.  Their goal is to discover 
if this transcript is correct.  The point is that if it is not, at least one player will notice it.  

Suppose that the first player speaks first.  She makes a measurement on the entangled 
state and determines that in the communication protocol she would speak a bit $b_1$.  
She then checks if $b_1$ agrees with $r_1$, the first bit of $r$.  Say that the second player 
speaks next.  Assuming that $r_1$ is the bit communicated by the first player, he then makes a 
measurement and determines a bit $b_2$ that he would communicate in the protocol.  
He then checks if $b_2$ agrees with $r_2$, the second bit of the random string.  
This process continues in this fashion as the players simulate the entire communication protocol.

If at any time player $i$ notices that a bit $r_t$ does not agree with what he would communicate,
assuming that the communication thus far has been given by $r_1 \cdots r_{t-1}$, we say that
$r$ is inconsistent with player $i$.  Otherwise it is consistent.

Now we define the output conditions
\begin{itemize}
  \item If $r$ is inconsistent with the first player, then she outputs a random bit in 
  $\{-1,+1\}$.  Otherwise, she outputs a bit $\{-1,+1\}$ with expectation $R[x_1, \ldots, x_N]$.
  \item If $r$ is inconsistent with player $i$ for $i>1$, then they output a random bit.  
  Otherwise, they output $1$.
\end{itemize}

Let $P[x_1, \ldots, x_N]$ be the expected output of this protocol on input $x_1, \ldots, x_N$.
Let us now compute the correlation of this protocol with $B$ under $\pi$:

\begin{align*}
\beta^*(B \circ \pi) &\ge \langle B \circ \pi , P \rangle \\
&=\frac{1}{2^c} \sum_{x_1, \ldots, x_N} \pi(x_1, \ldots, x_N) B[x_1, \ldots, x_N] R[x_1, \ldots, x_N] \\
&\ge \frac{1}{2^{c}} \left(\sum_{x_1, \ldots, x_N} \pi(x_1, \ldots, x_N) 
B[x_1, \ldots, x_N] A[x_1, \ldots, x_N]-2\eps \right)
\end{align*}

Rearranging, this gives the desired result:
$$
2^c \ge \max_{B,\pi} \frac{\langle A,B \circ \pi \rangle-2\eps}{\beta^*(B \circ \pi)}
$$
\end{proof}

In the two-party case, it is known that the model of shared entanglement and classical 
communication can simulate the model of shared entanglement and quantum communication
with a factor of two overhead.  The key idea is that if the parties share EPR-pairs, they can 
use these to pass quantum messages via teleportation with a cost of two classical bits 
per qubit.  We can also use this trick in the multiparty setting.
\begin{claim}
Let $A$ be a sign $N$-tensor. Let $\psiket$ be an entangled state, and let $R_\epsilon^{\psiket,E}(A)$ be
the minimum of $R_\epsilon^{\psiket'}(A)$ over all entangled states $\psiket'$ constituted of $\psiket$ together with an arbitrary 
number of EPR pairs. Then
$$
Q_\epsilon^{\psiket}(A) \ge \frac{R_\epsilon^{\psiket,E}(A)}{2}.
$$
\label{claim:tele}
\end{claim}

We now can prove Theorem~\ref{thm:cc}.

\begin{proof}[ of Theorem~\ref{thm:cc}]
Let $\psiket$ be a clique-wise entangled state containing $k$ coalitions.  We augment 
$\psiket$ to a state $\psiket'$ which additionally includes
an arbitrary number of shared EPR-pairs between each pair of players.  By 
Theorem~\ref{thm:entangledCC} and Claim~\ref{claim:tele} we have
$$
Q_\epsilon^{\psiket}(A) \ge \frac{1}{2}
\max_{B,\pi,\psiket'} \log \left( \frac{\langle A, B \circ \pi \rangle -2\epsilon}{\beta_{\psiket'}^*(B \circ \pi)} 
\right)
$$
where the maximum is over all sign $N$-tensors $B$, probability distributions $\pi$, and states $\psiket'$ constituted of $\psiket$ together with an arbitrary number of EPR pairs.
By Theorem~\ref{hyperghzbell}, for any such $B$, $\pi$ and $\psiket'$ we have
$$
\beta_{\psiket'}^*(B \circ \pi) \le 2^{3(k+N^2)N/2} \beta(B \circ \pi).
$$
Thus we obtain 
$$
Q_\epsilon^{\psiket}(A) \ge \frac{1}{2}
\max_{B,\pi} \log \left( \frac{\langle A, B \circ \pi \rangle -2\epsilon}{\beta(B \circ \pi)} 
\right) -O(kN^3).
$$
\end{proof}

Thus far we have phrased things for the NIH model of multiparty communication complexity.  
We can transfer this reasoning to the NOF model as follows.  For a function $f(x_1, \ldots, x_N)$ 
we can define a new function $f'$ that takes as arguments $N$ many $N-1$-tuples of 
strings.  We say that these tuples are consistent if they are a valid input to the NOF problem, 
that is if their union is exactly $N$ distinct strings.  When the arguments are consistent the value of 
$f'$ is the same as $f$, otherwise it is zero.  In the same way, for a probability distribution $\pi$ on 
$f$ we can define a distribution $\pi'$ on $f'$.  It can now be seen that 
$$
\beta(f' \circ \pi')=\max_{x_1\ldots, x_N}
\sum_{i_1, \ldots, i_N} (f\circ \pi)(i_1, \ldots, i_N) x_1(i_2, \ldots, i_N) \cdots x_N(i_1, \ldots, i_{N-1}),
$$
where the maximum is over functions $x_i:[n]^{N-1}\rightarrow\{-1,+1\}$.  The right-hand side is the standard definition 
of discrepancy in the number-on-the-forehead model (up to a constant $O(2^N)$ as discrepancy
is usually defined in terms of 0/1 vectors).  All our arguments carry through considering the 
function $f'$.  The fact that $f'$ is a much larger tensor than $f$ is immaterial as 
Grothendieck's inequality is independent of the size of the tensor.

Finally, we conclude this section by giving some examples of bounds that can be shown by the generalized discrepancy 
method.  Let $\GIP_n(x_1, \ldots, x_N)$ be the generalized inner product function, which returns 
the parity of the intersection size of the $x_i$.  Here the $x_i$ are $n$ bit strings.  
Babai, Nisan, and Szegedy showed a lower bound of $\tfrac{n}{2^{2N}}$ on the NOF complexity
of $\GIP_n$ using the discrepancy method \cite{babai:1992}.  The generalized discrepancy 
method can be used to show a bound of $\tfrac{n^{1/(N+1)}}{2^{2^N}}$ on the NOF complexity 
of the set intersection problem \cite{lee-shraibman:2009, chattopadhyay:2008}.
  

\section{Hardness of approximation of the entangled bias}\label{sec:approx}

As noted by Khot and Naor~\cite{khot:2008}, hardness of approximation results for Max-E3-Lin2 due to H\aa stad and Venkatesh~\cite{hastad:2004} can be extended to show that:
\begin{itemize}
\item Unless P=NP, for any constant $c>1$ there is no polynomial-time algorithm which approximates the classical bias of a three-party XOR game to within a multiplicative factor~$c$.
\item Unless NP$\subseteq$DTIME($n^{(\log n)^{O(1)}}$), for any $\eps>0$ the classical bias of a three-party XOR game cannot be approximated to within a multiplicative factor $2^{(\log n)^{1-\eps}}$ in time $2^{(\log n)^{O(1)}}$. 
\end{itemize}

The inapproximability results in~\cite{hastad:2004} only hold for \emph{symmetric} strategies, in which the players all share the same strategy. However, Khot and Naor show that the inapproximability result holds even when restricted to games $G=(A,\pi)$ that are invariant under permutations of the three players (i.e. $B[i,j,k] = B[i,k,j] = B[j,i,k] = B[j,k,i] = B[k,i,j] = B[k,j,i]$, where $B=A\circ\pi$) and are such that the same question is never asked to two players simultaneously (i.e. $B[i,j,j] = B[j,i,j] = B[j,j,i] = 0$).  In this case
Lemma~2.1 in~\cite{khot:2008} shows that the optimum with respect to symmetric strategies is within a factor $10$ of the general optimum. 

Combining this result with Theorems~\ref{schmidtbound} and~\ref{hyperghzbell} immediately gives a proof of 
Theorem~\ref{thm:inapprox}. Indeed, Theorem~\ref{schmidtbound} (resp. Theorem~\ref{hyperghzbell}) shows that, as long as the 
players are restricted to using an arbitrary Schmidt state (resp. clique-wise entanglement), the quantum bias is at most a constant 
times the classical bias. Hence any constant approximation to the quantum bias would give a 
constant approximation to the classical bias, which is ruled out by the hardness result 
from~\cite{hastad:2004}.


\section{Proof of Theorem~\ref{thm:qalgebra}}\label{app:varopoulos}

Here, we prove Theorem~\ref{thm:qalgebra}, which says that the Banach algebra formed by~$S_\infty$, the space of compact operators on a Hilbert space $\HS$, together with the Schur product (the entry-wise product), is a Q-algebra. 
The following theorem gives a simple characterization of a Q-algebra. It is a slight reformulation of a result by Davie \cite{davie:1973}, and is 
taken from Theorem~23 of \PG \cite{perezgarcia:2008}.

\begin{theorem}\label{thm:Q-characterization} Let $\X = (X,\cdot)$ be a commutative Banach algebra. Then~$\X$ is a Q-algebra if and only if there exists a universal constant $K$, such that for every choice of positive integers $N$ and $n$, $N$-tensor $A:[n]^N\to\R$, and functions $f_1,\dots,f_N:[n]\to \Ball(X)$, the following inequality holds:
\beq\label{eq:Q-characterization}
\Big\|\sum_{I\in[n]^N}A[I]f_1(i_1)\cdots f_N(i_N)\Big\|_{X} \leq K^N\|A\|_{\infty,\R},
\eeq
where $\|\cdot\|_{X}$ denotes the norm associated with the Banach space $X$.
\end{theorem}

\noindent Theorem~\ref{thm:Q-characterization} follows from the more standard characterization of Q-algebras of~\cite[Theorem~{18.7}]{diestel:1995} by using the inequality $\|A\|_{\infty,\C}\leq 2^N\|A\|_{\infty,\R}$, and the fact that without loss of generality, we may decouple the variables and consider $N$-linear forms instead of general polynomials~\cite[Lemma~{18.5} and~Proposition~{18.6}]{diestel:1995}.

%

\begin{proof}[ of Theorem~\ref{thm:qalgebra}]
We will show that Equation~\eqref{eq:Q-characterization} holds for $\X = (S_\infty,\circ)$. 
It follows from the Spectral Theorem that for any $\epsilon >0$, we can approximate any $T \in \Ball(S_\infty)$ by a finite-rank 
operator $T' \in S_{\infty}$. Since both $N$ and $n$ are finite, we only need to deal with a finite number of finite rank operators in $S_{\infty}$. All-together these operators act only on a finite-dimensional subspace of the original Hilbert space $\HS$. Hence, it will suffice to prove the statement for the case where $\HS$ is finite dimensional and the operators $f_l(i_l)$ are finite dimensional matrices. Setting $\epsilon = 1/(4N)$ introduces at most an extra factor of $2$ on the right-hand side of Equation~\eqref{eq:Q-characterization}. Further notice that it suffices to show this for Hermitian 
matrices $f_l(i_l)$, since for $T\in M_d$, we have that the matrix
$$
\begin{pmatrix}
0 & T \\
T^* & 0
\end{pmatrix}
$$
has the same norm as $T$ and is Hermitian.  We have
\begin{align}
\Big\|\sum_{I\in[n]^N} A[I] f_1(i_1) \circ \cdots \circ f_N(i_N) \Big\|_{S_\infty} &=
\max_{\alpha\in\Ball(\HS)} 
\Big| \bar{\alpha} \Big(\sum_{I\in[n]^N} A[I] f_1(i_1) \circ \cdots \circ f_N(i_N) \Big) 
\alpha \Big|\notag \\
&=\max_{\alpha\in\Ball(\HS)} \Big| 
\sum_{I\in[n]^N} A[I] \sum_{i,j} \bar \alpha_i \alpha_j \bra{i}^{\otimes N} 
f_1(i_1) \otimes \cdots \otimes f_N(i_N) \ket{j}^{\otimes N} \Big|,\label{eq:alphaeq}
\end{align}

\noindent where we used the fact that $\|\cdot\|_{S_{\infty}}$  simply denotes the spectral norm and wrote $\alpha = (\alpha_1,\alpha_2,\dots)$ using some orthonormal basis for $\HS$. Fix the $\alpha = (\alpha_1,\alpha_2,\dots)$ which maximizes this sum. 

Let $\ket{\Psi}=\sum_{i=1}^d \alpha_i \ket{i}^{\otimes N}$.  Then we can succinctly write the 
last expression in~\eqref{eq:alphaeq} as
\beq\label{qalgxor}
\left| 
\sum_{I\in[n]^N}A[I] \bra{\Psi} f_1(i_1) \otimes \cdots \otimes f_N(i_N) \ket{\Psi} \right|,
\eeq
where $\ket{\Psi} = \sum_i \alpha_i \ket{i}^{\otimes N}$. 
By the triangle inequality, replacing the $f_l(i_l)$ by $\pmset{}$-valued observables (Hermitian unitary matrices) which maximize the quantity~\eqref{qalgxor} can only increase its value since these observables are the extreme points in the convex set of Hermitian matrices of norm at most 1. Hence by definition,~\eqref{qalgxor} is bounded by the bias $\betad(A)$. Theorem~\ref{schmidtbound} then implies 
that Equation~\eqref{eq:Q-characterization} holds with a constant $K=2^{3/2}$. 
\end{proof}

For further information on this problem, we refer to~\cite{varo:1975,lemerdy:1998,perez:2006} and for  information on Q-algebras, we refer to \cite[Chapter 18]{diestel:1995}.


\section{Grothendieck-type inequalities}\label{statestofunctionals}

 We prove the following tri-linear extension of Grothendieck's inequality:

\begin{theorem}\label{stabilizerform}
Let $\K = \R$ or $\C$, let $G = (V,E)$ be a simple undirected graph and $(V_1,V_2,V_3)$ be a partitioning of $V$. For each $x\in V_l$ let $\HS(x,l)$ be a two-dimensional Hilbert space with underlying field $\K$ and $\HS_l = \bigotimes_{x\in V_l}\HS(x,l)$. Define the linear functional $\Phi_G:\bigotimes_{l=1}^3\HS_l\to\K$ by
\beqn
\Phi_G:\bigotimes_{l=1}^3v_l\mapsto \sum_{S_1\subseteq V_1}\sum_{S_2\subseteq V_2}\sum_{S_3\subseteq V_3}(-1)^{|E(S_1\cup S_2\cup S_3)|}\prod_{l=1}^3v_l\big(S_l\big),
\eeqn
where $v_l\in\HS_l$ and we index the $2^{|I_l|}$ coordinates of $v_l$ by the subsets $S_l\subseteq I_l$. Then $\Phi_G$ satisfies that for every 3-tensor $A:[n]^3\to\K$ and set of functions $f_l:[n]\to\Ball(\HS_l)$, the following inequality holds
\beqn
\Big|\sum_{i,j,k=1}^nA[i,j,k]\Phi_G\big(f_1(i)\otimes f_2(j)\otimes f_3(k)\big)\Big| \leq C\, \|A\|_{\infty,\K}
\eeqn
where $C = O(2^{|V|/2})$.
\end{theorem}

To prove this, we use the following theorem of Bravyi et al. \cite{bravyi:2006}.

\begin{theorem}\label{bravyithm}
Let $\HS_1,\HS_2,\HS_3$ be complex Hilbert spaces and $\ket{\Psi}\in\HS_1\otimes\HS_2\otimes\HS_3$ a stabilizer state. Then there exist unitary operators $U_1,U_2,U_3$ on $\HS_1,\HS_2,\HS_3$, respectively, such that the state $U_1\otimes U_2\otimes U_3\ket{\Psi}$ is equal to a collection of GHZ and Bell states.
\end{theorem}

The linear functionals appearing in Theorem~\ref{stabilizerform} are derived from a special class of stabilizer states known as graph states. A $q$-qubit graph state is a unit vector in $\C^{2^q}$ which is uniquely defined by a simple  undirected graph $G = (V,E)$ on $q$ vertices.  The graph state associated with $G$ is given by
\beq\label{graphstate}
\ket{\Psi} = \frac{1}{\sqrt{2^q}}\sum_{S\subseteq V}(-1)^{|E(S)|}\ket{S},
\eeq
where $|E(S)|$ denotes the number of edges in the subgraph of $G$ induced by the vertices in $S$ and~$\ket{S}$ denotes the computational basis state corresponding to the length-$q$ characteristic vector of the set~$S$.

\begin{proof}[ of Theorem~\ref{stabilizerform}]

Let $|V|=q$ and let $\ket{\Psi}$ be the unique graph state associated with $G$, as given by \eqref{graphstate}. 
Let $(V_1,V_2,V_3)$ be a partitioning of $V$ such that party $l$ has the qubits indexed by the labels in $V_l$ and denote the respective Hilbert spaces by $\HS_1$, $\HS_2$ and $\HS_3$. Then, using an appropriate arrangement of the Hilbert spaces, we may write $\ket{\Psi}$ as
\beqn
\ket{\Psi} = \frac{1}{\sqrt{2^q}}\sum_{S_1\subseteq V_1}\sum_{S_2\subseteq V_2}\sum_{S_3\subseteq V_3}(-1)^{|E(S_1\cup S_2\cup S_3)|}\ket{S_1}\otimes \ket{S_2}\otimes \ket{S_3},
\eeqn

It is easily seen that for vectors $v_l\in\HS_l$ and linear functional $\Phi_G$ as defined in the theorem, we have $(v_1\otimes v_2\otimes v_3)\cdot \ket{\Psi} = 2^{-q/2}\Phi_G(v_1\otimes v_2\otimes v_3)$.

Now let $U_1$, $U_2$ and $U_3$ be the unitary operators on $\HS_1$, $\HS_2$ and $\HS_3$, respectively, from Theorem~\ref{bravyithm}. Then, by the unitary invariance of the sets $\Ball(\HS_l)$, we have that, for any 3-tensor $A:[n]^3\to\R$ and $f_l:[n]\to\Ball(\HS_l)$,

\begin{align}
\Big| \sum_{i,j,k}A[i,j,k]\, &\Phi_G\big(f_1(i)\otimes f_2(j)\otimes f_3(k)\big) \Big| = 2^{q/2}\Big| \sum_{i,j,k}A[i,j,k]\, (f_1(i)\otimes f_2(j)\otimes f_3(k))\cdot\ket{\Psi} \Big|\notag \\
&\leq 2^{q/2}\max_{g_l:[n]\to\Ball(\HS_l)}\Big| \sum_{i,j,k}A[i,j,k]\, (g_1(i)\otimes g_2(j)\otimes g_3(k))\cdot\ket{\Psi} \Big|\notag\\
&\leq 2^{q/2}\max_{g_l':[n]\to\Ball(\HS_l)}\Big| \sum_{i,j,k}A[i,j,k]\, (U_1g_1'(i)\otimes U_1g_2'(j)\otimes U_3g_3'(k))\cdot\ket{\Psi} \Big|\notag\\
&= 2^{q/2}\max_{g_l':[n]\to\Ball(\HS_l)}\Big| \sum_{i,j,k}A[i,j,k]\, g_1'(i)\otimes g_2'(j)\otimes g_3'(k)\cdot (U_1\otimes U_2\otimes U_3)\ket{\Psi} \Big| \notag\\
&\leq 2^{q/2}\max_{\ket{\Psi_C}}\max_{h_l:[n]\to\Ball(\HS_l)}\Big| \sum_{i,j,k}A[i,j,k]\, h_1(i)\otimes h_2(j)\otimes h_3(k)\cdot \ket{\Psi_C} \Big|,\label{stabcliqueineq}
\end{align}
where $\ket{\Psi_C}$ is a clique-wise entangled state, shared among the three parties. For a hypergraph $H = (T,E')$ with vertex set $T = [3]$ on three vertices and some dimension $d$, such a state has the form
\beqn
\ket{\Psi_C} = \frac{1}{\sqrt{d^{|E'|}}}\sum_{J\in[d]^{E'}}\bigotimes_{l=1}^3\ket{J_{|E'(l)}}
\eeqn
as we saw in the proof of Theorem~\ref{hyperghzbell}. Now, by Claim~\ref{claim:phimap}, we have that for vectors $a\in\HS_1, b\in\HS_2,c\in\HS_3$, the following equalities hold:
\beqn
(a\otimes b\otimes c)\cdot \ket{\Psi_C} = \frac{1}{\sqrt{d^{|E'|}}}\sum_{J\in[d]^{E'}}a(J_{|E'(1)})\cdot b(J_{|E'(2)})\cdot c(J_{|E'(3)}) = \frac{1}{\sqrt{d^{|E'|}}}\widetilde{\Phi}(a\otimes b\otimes c)
\eeqn
where $\widetilde{\Phi} = \big(\bigotimes_{e\in E'}\psi_e\big)\circ\sigma$ for generalized inner-product functions $\psi_e$ on each of the edges. Hence, the last expression in \eqref{stabcliqueineq} is equal to
\beqn
\sqrt{\frac{2^{q}}{d^{|E'|}}}\max_{h_l:[n]\to\Ball(\HS_l)}\Big| \sum_{i,j,k}A[i,j,k]\, \widetilde{\Phi}\big(h_1(i)\otimes h_2(j)\otimes h_3(k)\big) \Big|.
\eeqn
The result now follows directly from Theorems~\ref{kgrothbound}~and~\ref{carnethm} and the fact that $d\geq 1$.
\end{proof}

\section*{Acknowledgements}JB thanks Peter H\o yer, TL thanks Gideon Schechtman, and TV thanks Falk Unger for useful discussions.

\bibliographystyle{alphaabbrv}
\bibliography{newref}

\appendix

\section{Proof of Tonge's theorem}\label{apptonge}

In this section we prove an extension of Theorem~\ref{kgrothbound} to the case where the base Hilbert space is not necessarily complex:

\begin{theorem}\label{kgrothbound2}
Let $n,N\geq 2$ and $d$ be positive integers, $\K\in\{\R,\C\}$, and $\HS$ be a $d$-dimensional Hilbert space with underlying field $\K$. Then, for every $N$-tensor $A:[n]^{N}\to\K$ and $f_1,\dots,f_N:[n]\to \Ball(\HS)$, the following inequality holds:
\beq\label{tongesthm}
\Big|\sum_{i_1,\dots,i_N=1}^nA[i_1,\dots,i_N]\langle f_1(i_1),\dots,f_k(i_N)\rangle\Big| \leq 
2^{(N-2)/2}K_G^{\K}\|A\|_{\infty,\K},
\eeq
where $K_G^{\R}$ and $K_G^{\C}$  are the real and complex Grothendieck constant, respectively. Moreover, if the underlying field for $A$ and the scalars on the right-hand side is $\R$, but the underlying field for $\HS$ is~$\C$, then the inequality 
\beq\label{tongeextended}
\Big|\sum_{i_1,\dots,i_N=1}^nA[i_1,\dots,i_N]\langle f_1(i_1),\dots,f_k(i_N)\rangle\Big| \leq 
2^{(3N-5)/2}K_G^{\C}\|A\|_{\infty,\R}
\eeq
holds.
\end{theorem}

Tonge proved inequality~\eqref{tongesthm} and Theorem~\ref{kgrothbound} is inequality~\eqref{tongeextended}. 

\begin{remark}
Note that from this extension, it follows that if the players use a real state and observables which can be represented by real matrices, then Theorems~\ref{schmidtbound} and~\ref{hyperghzbell} are valid with the constants on the right-hand side replaced by $2^{(N-2)/2}K_G^{\R}$ and $2^{k(r-2)/2}(K_G^{\R})^k$, respectively. In particular, this implies that Zukowski's \CQ-gap for the games given in~\cite{zukowski:1993} cannot be achieved with a strategy that involves a real Schmidt state and only real observables.
\end{remark}

The proof of Theorem~\ref{kgrothbound2} is by induction on $N$ and uses the following slight modification of a Theorem by Littlewood \cite{littlewood:1930} (see also \cite[page 43]{pietsch:1972} and \cite{szarek:1976}) in the inductive step.

\begin{lemma}[Slight extension of Littlewood 1964]\label{littlewood}
Let $n,d$ be positive integers and $\K\in\{\R,\C\}$. Then, for every $n\times d$ matrix $M:[n]\times[d]\to\K$, the following inequality holds:
\beq\label{littleineq}
\sum_{i=1}^n\Big(\sum_{j=1}^d|M[i,j]|^2\Big)^{1/2} \leq \sqrt{2}\mathop{\max_{\phi\in\Ball(\ell_{\infty}^n)}}_{\chi\in\Ball(\ell_{\infty}^d)}\Big|\sum_{i=1}^n\sum_{j=1}^dM[i,j]\phi(i)\chi(j)\Big|,
\eeq
where the underlying field for $\phi$ and $\chi$ is $\K$. Moreover, if the underlying field for $M$ and $\chi$ is~$\C$, but that for $\phi$ is $\R$, then the inequality
\beq\label{littleineqextended}
\sum_{i=1}^n\Big(\sum_{j=1}^d|M[i,j]|^2\Big)^{1/2} \leq 2\sqrt{2}\mathop{\max_{\phi\in\pmset{n}}}_{\chi\in\Ball(\ell_{\infty}^d)}\Big|\sum_{i=1}^n\sum_{j=1}^dM[i,j]\phi(i)\chi(j)\Big|,
\eeq
holds.
\end{lemma}

Inequality~\eqref{littleineq} is Littlewood's inequality.
These ineqaulities, in turn, can be derived from Khintchine's inequality, which states that for $0<p<\infty$, there exist constants $A_p$ and $B_p$ such that for every finite sequence of real or complex scalars $(c_j)_{i=1}^n$ the following inequality holds: 
\beq\label{khintchine}
A_p\left(\sum_{i=1}^n|c_i|^2\right)^{1/2} \leq \left(\int_{t=0}^1\Big|\sum_{i=1}^nc_ir_i(t)\Big|^pdt\right)^{1/p}\leq B_p\left(\sum_{i=1}^n|c_i|^2\right)^{1/2},
\eeq
where $r_i(t) = \sign\big(\sin (2^i\pi t)\big)$ denotes the $i$'th Rademacher function. Haagerup~\cite{haagerup:1982} found all values of $A_p$ and $B_p$ for sequences of real numbers $c_i$. The best value of $A_1$ is due to Szarek~\cite{szarek:1976}, who proved that $A_1=1/\sqrt{2}$ for sequences of real scalars (see \cite{tomaszewski:1987} for an elementary proof). He also has an argument attributed to Tomaszewski which implies that for sequences of complex scalars and $p=1$, the left-hand side of~\eqref{khintchine} holds with $A_1\geq 1/\sqrt{2}$.
 

\begin{proof}[ of Lemma \ref{littlewood}]
If we set $p = 1$ and use the left side of Khintchine's inequality~\eqref{khintchine} for every $i=1\ldots n$, we get 
\beqrn
\sum_{i=1}^n\Big(\sum_{j=1}^d|M[i,j]|^2\Big)^{1/2} &\leq& \sqrt{2}\int_{t=0}^1\sum_{i=1}^n\Big|\sum_{j=1}^dM[i,j]r_j(t)\Big|dt\\
&\leq& \sqrt{2}\sup_{x:[d]\to\pmset{}}\Big\{\sum_{i=1}^n\Big|\sum_{j=1}^dM[i,j]x(j)\Big|\Big\}.
\eeqrn

Inequality~\eqref{littleineq} now follows from the fact that there exists $\phi\in\Ball(\ell_{\infty}^n)$ and $\chi\in\Ball(\ell_{\infty}^d)$ such that
\beqn
\Big|\sum_{i=1}^n\sum_{j=1}^dM[i,j]\phi(i)\chi(j)\Big| = \sum_{j=1}^d\Big|\sum_{i=1}^nM[i,j]x(j)\Big|.
\eeqn

To see this, set $\chi(j) = x(j)$ and $\phi(i) = \big(\sum_{j=1}^dM[i,j]x(j)\big)^*/\big|\sum_{j=1}^dM[i,j]x(j)\big|$ for every $i\in[n]$ and $j\in[d]$. 

To prove inequality~\eqref{littleineqextended}, consider the case $\K=\C$ and let $\chi$ and $\phi$ be the complex sequences that maximize the right-hand side of inequality~\eqref{littleineq}. We write this quantity as the inner product $|\phi\cdot a|$, where $a_i = \sum_{j=1}^dM[i,j]\chi(j)$. By the triangle inequality, we have
\beqr
|\phi\cdot a| &=& \big|\Re(\phi)\cdot a + i\Im(\phi)\cdot a\big|\nonumber\\
&\leq& \big|\Re(\phi)\cdot a\big| + \big|\Im(\phi)\cdot a\big|\nonumber\\
&\leq& 2\max\big\{\big|\Re(\phi)\cdot a\big|, \big|\Im(\phi)\cdot a\big|\}.\label{littlerealcomplexpf}
\eeqr
Without loss of generality, we may assume that this maximum is achieved with $\phi' = \Re(\phi)$. This is a vector in $[-1,1]^n$ for which the inequality
\beqn
\sum_{i=1}^n\Big(\sum_{j=1}^d|M[i,j]|^2\Big)^{1/2} \leq 2\sqrt{2}\Big|\sum_{i=1}^n\sum_{j=1}^dM[i,j]\chi(j)\phi'(i)\Big|
\eeqn
holds. By convexity, we have that the maximum is achieved at one of the extreme points of $[-1,1]^n$. This completes the proof.
\end{proof}

We now turn to the proof of Theorem~\ref{kgrothbound2}.

\begin{proof}[ of Theorem \ref{kgrothbound2}]
(By induction on $N$.) For the base case, $N=2$, we only need to prove something for inequality~\eqref{tongeextended}, since the base cases for inequality~\eqref{tongesthm} are the  real and complex Grothendieck inequality. We use the following simplified version of \cite[Proposition 15]{munoz:1999}:

\begin{claim}\label{claim:mst99}
 For a real matrix $A[i,j]$, we have
\beq\label{absolutereal}
\max_{\alpha_i,\beta_j\in\Ball(\C)}\Big|\sum_{i,j}A[i,j]\alpha_i\beta_j\Big| \leq \max_{\alpha_i',\beta_j'\in\Ball(\C)}\sum_{i,j}A[i,j]\Re(\alpha_i'\beta_j'),
\eeq
\end{claim}

\begin{proof}
For a complex number $\gamma:= re^{i\phi}$ (with $r\geq 0$) we have that $e^{-i\phi} \gamma = r = |\gamma|$. Trivially, we have that $\Re(e^{-i\phi} \gamma) = e^{-i\phi} \gamma$. Hence, for complex $\alpha_i,\beta_j$ and real $A_{ij}$, we have that there exists a $\phi$ such that
\beqn
\Re\Big(e^{-i\phi}\sum_{i,j}A[i,j]\alpha_i\beta_j\Big) = e^{-i\phi}\sum_{i,j}A[i,j]\alpha_i\beta_j = \Big|\sum_{i,j}A[i,j]\alpha_i\beta_j\Big|.
\eeqn
%

Let $\alpha_i,\beta_j\in\Ball(\C)$ that maximize the left-hand side of \eqref{absolutereal}. We have that for some $\phi$:
\beqrn
\Big|\sum_{i,j}A[i,j]\alpha_i\beta_j\Big| &=& \Re\Big(e^{-i\phi}\sum_{i,j}A[i,j]\alpha_i\beta_j\Big)\\
&=&\sum_{i,j}A[i,j]\Re\Big((\alpha_ie^{-i\phi/2})(\beta_je^{-i\phi/2})\Big)\\
&\leq& \max_{\alpha_i',\beta_j'\in\Ball(\C)}\sum_{i,j}A[i,j]\Re(\alpha_i'\beta_j').
\eeqrn
\end{proof}

We can write the real part $\Re(\alpha_i\beta_j)$ of two complex numbers $\alpha_i,\beta_j$ as the inner product between real vectors $a_i = \big(\Re(\alpha_i),\Im(\alpha_i)\big)^T$ and $b_j = \big(\Re(\beta_j),-\Im(\beta_j)\big)^T$. Using this and Claim~\ref{claim:mst99}, we get that for every sequence of unit vectors $u_i,v_j\in\Ball(\C^d)$,
\beqrn
\Big|\sum_{i,j=1}^{n}A[i,j] u_i\cdot v_j\Big|
 &\leq& K_G^{\C}\max_{\alpha_i,\beta_j\in\Ball(\C)}\Big|\sum_{i,j=1}^{n}A[i,j]\alpha_i\beta_j\Big|\\
&\leq& K_G^{\C}\max_{\alpha_i,\beta_j\in\Ball(\C)}\sum_{i,j=1}^{n}A[i,j]\Re(\alpha_i\beta_j)\\
 &\leq&K_G^{\C}\max_{a_i,b_j\in\Ball(\R^2)}\sum_{i,j=1}^{n}A[i,j]a_i\cdot b_j\\
&\leq& K_G^{\R}(2)K_G^{\C}\|A\|_{\infty,\R} =\sqrt{2}K_G^{\C}\|A\|_{\infty,\R},
\eeqrn
where we used the fact that $K_G^{\R}(2) =\sqrt{2}$, as Krivine showed.
This proves the base case.

Define the $n\times d$ matrix:
\beqn
B[i_N,j] := \sum_{i_1,\dots,i_{N-1}=1}^nA[i_1,\dots,i_{N-1},i_k]f_1(i_1)_j\cdots f_{N-1}(i_{N-1})_j.
\eeqn

By the triangle inequality, the Cauchy-Schwarz inequality and inequality~\eqref{littleineqextended}, we have
\beqr
\Big|\sum_{i_1,\dots,i_N=1}^nA[i_1,\dots,i_N]\langle f_1(i_1),\dots,f_N(i_N)\rangle\Big| &=&  \Big|\sum_{i_N=1}^n\sum_{j=1}^dB[i_N,j]f_{N}(i_N)_j\Big|\nonumber\\
&\leq& \sum_{i_N=1}^n\Big|\sum_{j=1}^dB[i_N,j]f_{N}(i_N)_j\Big|\nonumber\\
&\leq& \sum_{i_N=1}^n\Big(\sum_{j=1}^d\big|B[i_N,j]\big|^2\Big)^{1/2}\nonumber\nonumber\\
&\leq& 2\sqrt{2}\mathop{\max_{\phi:[n]\to\pmset{}}}_{\chi:[d]\to\Ball(\C)}\Big|\sum_{i_N=1}^n\sum_{j=1}^dB[i_N,j]\phi(i_N)\chi(j)\Big|\label{lastcauschwlit}
\eeqr

Let $\phi^*:[n]\to\pmset{}$ be the function that maximizes \eqref{lastcauschwlit}. Applying the induction hypothesis to the real $(N-1)$-tensor
\beqn
C[i_1,\dots,i_{N-1}]:=\sum_{i_N=1}^n A[i_1,\dots,i_{N-1},i_N]\phi^*(i_N)
\eeqn
allows us to bound \eqref{lastcauschwlit} as follows:
\begin{multline*}
2\sqrt{2}\max_{\chi\in\Ball(\ell_{\infty}^d)}\Big|\sum_{i_N=1}^n\sum_{j=1}^dB[i_N,j]\phi^*(i_N)\chi(j)\Big| \\
=2\sqrt{2}\max_{\chi\in\Ball(\ell_{\infty}^d)}\Big|\sum_{i_N=1}^n\sum_{j=1}^d\Big(  \sum_{i_1,\dots,i_{N-1}=1}^nA[i_1,\dots,i_{N-1},i_N]f_1(i_1)_j\cdots f_{N-1}(i_{N-1})_j   \Big)\phi^*(i_N)\chi(j)\Big| \\
 =2\sqrt{2}\max_{\chi\in\Ball(\ell_{\infty}^d)}\Big|\sum_{i_1,\dots,i_{N-1}=1}^nC[i_1,\dots,i_{N-1}]\cdot \Big(\sum_{j=1}^d\chi(j) f_1(i_1)_j\cdots f_{N-1}(i_{N-1})_j\Big)\Big| \\
\leq2\sqrt{2}\cdot \big((\sqrt{2}\,(2\sqrt{2})^{((N-1)-2)}\big)K_G^{\C} \max_{\phi_1,\dots,\phi_{N-1}\in\pmset{n}}\Big|\sum_{i_1,\dots,i_{N-1}=1}^n C[i_1,\dots,i_{N-1}]\phi_1(i_1)\cdots \phi_{N-1}(i_{N-1})\Big| \\
 = \sqrt{2}\,2^{3(N-2)/2}K_G^{\C} \max_{\phi_1,\dots,\phi_{N-1}\in\pmset{n}}\Big|\sum_{i_1,\dots,i_{N-1}=1}^n\Big(\sum_{i_N=1}^n A[i_1,\dots,i_{N-1},i_N]\phi^*(i_N)\Big)\phi_1(i_1)\cdots \phi_{N-1}(i_{N-1})\Big|\\
\leq \sqrt{2}\,(2\sqrt{2})^{3(N-2)/2}K_G^{\C}\max_{\phi_1,\dots,\phi_N\in\pmset{n}}\Big|\sum_{i_1,\dots,i_N=1}^n A[i_1,\dots,i_{N-1},i_N]\phi_1(i_1)\cdots \phi_N(i_N)\Big|.
\end{multline*}
The final term is $2^{(3N-5)/2}K_G^{\C}\|A\|_{\infty,\R}$. This proves inequality~\eqref{tongeextended}. Inequality~\eqref{tongesthm} is proved in the same way, except with the original Grothendieck inequality for the base case and Littlewood's inequality~\eqref{littleineqextended} in Equation~\eqref{lastcauschwlit}, giving the factor $2^{(N-2)/2}K_G^{\K}$. This completes the proof.
\end{proof}

\section{Proof of Carne's theorem}\label{app:carne}

In this section we prove Theorem~\ref{carnethm}.

\begin{proof}[ of Theorem~\ref{carnethm}]
The proof is by induction on the number of edges $|E|$. If the edge set is empty, then there is nothing to prove. Let $e_0$ be any edge in $G$, and consider the graph $G_0 = (V,E\backslash\{e_0\})$. To re-write the expression, first assume that each vector $f_x(i_x)\in\HS_x = \bigotimes_{e\in E(x)} \HS(x,e)$ has the following tensor structure: 
$$f_x(i_x) = f_x^0(i_x)\otimes f_x^1(i_x)$$
where $f_x^0(i_x)\in\otimes_{e\in E\backslash\{e_0\}} \HS(x,e)$ and $f_x^1(i_x)\in \HS(x,e_0)$. Define $\Phi_{G_0} = \left(\bigotimes_{e\in E\backslash\{e_0\}} \psi_e\right)\circ \sigma_{G_0}$, where $\sigma_{G_0}$ is the re-arranging map for $G_0$. With this notation we have
\begin{align*}
\Phi\Big(\bigotimes_{x\in V}f_x(i_x)\Big) &= \Phi\Big(\bigotimes_{x\in V} f_x^0(i_x) \otimes f_x^1(i_x) \Big) \\
&= \Phi_{G_0}\Big(\bigotimes_{x\in V}f_x^0(i_x)\Big) \cdot \psi_{e_0} \Big( \bigotimes_{x\in e_0} f_x^1(i_x)\Big)
\end{align*}
Define the tensor $B[I] = A[I] \cdot \psi_{e_0} \left( \bigotimes_{x\in e_0} f_x^1(i_x)\right)$. Applying the induction hypothesis to $B[I]$ and the graph $G_0$ (note that the $\psi_{e_0}(\cdots)$ term is simply a number, dependent on $I$) gives
\begin{align}\label{eq:first}
\sum_{I\in[n]^V}B[I]\cdot \Phi_{G_0}\Big(\bigotimes_{x\in V}f_x^0(i_x)\Big) \leq \Big(\prod_{e\in E\backslash\{e_0\}} C_e^{\K'}\Big)\|B\|_{\infty,\K'}
\end{align}
By definition, 
\begin{align*}
\|B\|_{\infty,\K'} &= \max_{\phi_1,\ldots,\phi_N \in \Ball(l_\infty^n)} \left| \sum_I B[I]\, \phi_1(i_1)\cdots \phi_N(i_N)\right|\\
&= \max_{\phi_1,\ldots,\phi_N \in \Ball(l_\infty^n)} \left| \sum_I A[I] \,\phi_1(i_1)\cdots \phi_N(i_N) \,\psi_{e_0} \left( \bigotimes_{x\in e_0} f_x^1(i_x)\right)\right|
\end{align*}
Fix $\phi_1,\ldots,\phi_N$ that achieve this maximum, and define the tensor $C[I] = A[I] \phi_1(i_1)\cdots \phi_N(i_N)$. By hypothesis, the function $\psi_e$ enjoys a Grothendieck-type inequality, hence the expression above can be bounded by
\begin{align}\label{eq:second}
\|B\|_{\infty,\K'}\,=\,\sum_I C[I]\cdot \psi_{e_0} \left( \bigotimes_{x\in e_0} f_x^1(i_x)\right) \,\leq\, C_{e_0}^{\K'} \|C\|_{\infty,\K'}
\end{align}
To conclude, we  can relate $\|C\|_{\infty,\K'}$ to $\|A\|_{\infty,\K'}$ in the following way:
\begin{align*}
\|C\|_{\infty,\K'} &= \max_{\phi'_1,\ldots,\phi'_N \in \Ball(l_\infty^n)} \left| \sum_I C[I]\, \phi'_1(i_1)\cdots \phi'_N(i_N)\right|\\
&= \max_{\phi'_1,\ldots,\phi'_N \in \Ball(l_\infty^n)} \left| \sum_I A[I]\, \phi_1(i_1)\phi'_1(i_1)\cdots \phi_N(i_N) \phi'_N(i_N)\right|\\
&= \max_{\phi''_1,\ldots,\phi''_N \in \Ball(l_\infty^n)} \left| \sum_I C[I]\, \phi''_1(i_1)\cdots \phi''_N(i_N)\right|\\
&= \|A\|_{\infty,\K'}
\end{align*}
Combining Eqs.~(\ref{eq:first}) and~(\ref{eq:second}) gives the result in the case where all $f_x(i_x)$ have the tensor structure we described earlier. If not, since $\Phi$ is linear, writing their Schmidt decomposition will result in a weighted sum of expressions involving only unit vectors of this form. The weighted sum can be bounded by its maximum component, for which we can apply the reasoning above.
 \end{proof}

\end{document}